\documentclass{article}
\pdfoutput=1
\usepackage{arxiv}
\PassOptionsToPackage{numbers,sort, compress}{natbib}

\usepackage[utf8]{inputenc}
\usepackage[T1]{fontenc} 
\usepackage{geometry}
\usepackage{xcolor}
\usepackage{natbib}
\usepackage{caption}
\usepackage{amsmath}
\usepackage{amssymb}
\usepackage[ruled,vlined]{algorithm2e}
\usepackage{bbold}
\usepackage{mathtools}
\usepackage{etoolbox}
\usepackage{xcolor}
\usepackage[english]{babel}
\usepackage{amsthm}
\usepackage{hyperref}
\usepackage{subcaption}
\usepackage{float}
\usepackage{paralist}
\usepackage{bbold}
\usepackage{url}            
\usepackage{booktabs}       
\usepackage{amsfonts}       
\usepackage{nicefrac}       
\usepackage{microtype}      
\usepackage{bm}
\usepackage[page]{appendix}
\usepackage[inline]{enumitem}
\usepackage{algorithmic}

\newtheorem{observation}{Observation}

\newtheorem{fact}{Fact}

\newtheorem{proposition}{Proposition}

\newtheorem{definition}{Definition}
\newtheorem{lemma}{Lemma}
\newtheorem{theorem}{Theorem}
\newtheorem*{theorem*}{Theorem}
\newtheorem{corollary}{Corollary}

\newcommand{\Tr}{\mathrm{Tr}}
\newcommand{\id}{\mathbb{1}}

\DeclarePairedDelimiter{\ceil}{\lceil}{\rceil}
\hypersetup{
    colorlinks=true,
    linkcolor = black,
    urlcolor=blue,
    citecolor=black
}

\title{Matrix Multiplicative Weights Updates in Quantum Zero-Sum Games: Conservation Laws \& Recurrence}
\author{Rahul Jain\thanks{Centre for Quantum Technologies and Department of Computer Science, National University of Singapore and MajuLab, UMI3654, Singapore.}\\
  \texttt{rahul@comp.nus.edu.sg}\\
  \And
  Georgios Piliouras\thanks{Engineering Systems and Design, Singapore University of Technology and Design, Singapore} \\
  \texttt{georgios@sutd.edu.sg} \\
  \AND
  Ryann Sim$^\dagger$ \\
  \texttt{ryann\_sim@mymail.sutd.edu.sg} \\
  }

\date{}
\renewcommand{\headeright}{}
\renewcommand{\undertitle}{}

\begin{document}

\maketitle


\begin{abstract}
Recent advances in quantum computing and in particular, the introduction of quantum GANs, have led to increased interest in quantum zero-sum game theory, extending the scope of learning algorithms for classical games into the quantum realm. In this paper, we focus on learning in quantum zero-sum games under \emph{Matrix Multiplicative Weights Update} (a generalization of the multiplicative weights update method) and its continuous analogue, \emph{Quantum Replicator Dynamics}. When each player selects their state according to quantum replicator dynamics, we show that the system exhibits conservation laws in a quantum-information theoretic sense. Moreover, we show that the system exhibits Poincar\'{e} recurrence, meaning that almost all orbits return arbitrarily close to their initial conditions infinitely often. Our analysis generalizes previous results in the case of classical games~\cite{Piliouras2014OptimizationDC,mertikopouloscycles}.

\end{abstract}

\section{Introduction}
\label{sec:intro}

The Nash equilibrium has always been a central concept in non-cooperative game theory. While the existence of Nash equilibria is well known in finite games where mixed strategies are allowed \cite{nash1951games}, it is less clear how to efficiently compute these equilibria themselves. Indeed, much work revolves around finding methods to solve for Nash equilibria in different contexts, as well as analyzing the complexity of solving for the Nash \cite{lemkehowson, daskalakis2009complexity}. 

Despite the dominance of equilibrium computation in classical theory, recent works have begun to move towards attempting to understand the nature of learning in games \cite{fudenberg1998theory, sandholm2010population}. The family of regret-minimizing algorithms known as Follow-The-Regularized-Leader (FTRL) have been studied extensively, both for discrete time \cite{bailey2018multiplicative, NEURIPS2018_90e13578, palaiopanos2017multiplicative, piliouras2021optimal} and continuous time \cite{mertikopouloscycles, flokas2019poincare, Piliouras2014OptimizationDC, perolat2020poincare} dynamics. These algorithms have found many applications in the domain of machine learning - one such example is that of Generative Adversarial Networks (GANs) \cite{goodfellow2014generative}, which have been successfully modeled using zero-sum game theory. This rich connection has led to a stronger understanding of adversarial learning and improved practical performance of GANs in various settings \cite{flokas2019poincare, kodali2017convergence, lin2020gradient}.

On the other side of the coin, quantum computing is a field which has garnered much interest in the machine learning community. Many quantum machine learning algorithms have been proposed, extending standard classical algorithms. One such example is quantum GANs \cite{dallaire2018quantum,lloyd2018quantum, chakrabarti2019quantum}. While such works describe the architecture of quantum GANs, we still do not fully understand the behavior of learning algorithms in the quantum setting. In this paper, we provide an initial analysis of learning in zero-sum quantum games. We do so by studying the learning behaviour of the ubiquitous Multiplicative Weights Update (MWU) update rule \cite{littlestone1994weighted, freund1999adaptive, arora_multiplicative}, a specific instance of the FTRL framework, when applied to such games. 


There have been several important results studying MWU and its continuous counterpart, \emph{replicator dynamics}. \cite{bailey2018multiplicative} showed that classical MWU does not converge in day-to-day behaviour to the mixed Nash equilibrium when applied to two player zero-sum games. \cite{mertikopouloscycles} showed that the orbits of players in zero-sum games exhibit Poincar\'{e} recurrence when using replicator dynamics. However, quantum systems can oftentimes behave in radically different ways to their classical counterparts. Indeed, allowing for quantum strategies can create situations wherein players can gain greater payoffs than if they were playing using only classical strategies (see e.g. the Mermin–Peres magic square game \cite{mermin1990simple,peres1990incompatible}). As such, our extension of the classical results to the quantum realm has to be treated with care.

We focus on a matrix generalization of MWU known as \emph{Matrix Multiplicative Weights Update} (MMWU), which changes the paradigm of standard MWU from cost vectors to cost matrices, and from probability vectors to density matrices. This generalization has been independently discovered and studied by \cite{tsuda2005matrix} as Matrix Exponentiated Gradient Updates and \cite{arora_multiplicative} as the Matrix Multiplicative Weights algorithm. Applications of the MMWU algorithm include solving semi-definite programs (SDPs) \cite{arora2016combinatorial} and obtaining bounds on the sample complexity for learning problems in quantum computing \cite{kale2007efficient}. From a game theoretic standpoint, \cite{jain2009parallel} analyzed the MMWU algorithm for quantum zero-sum games, proving time-average convergence to an approximate Nash equilibrium. We extend this result by studying the learning behaviour of MMWU in quantum zero-sum games, which could lay the groundwork for developing novel algorithms that achieve convergence to the Nash in day-to-day behaviour in quantum zero-sum games, similar to the case of their classical analogues~\cite{daskalakis2018training,mertikopoulos2019optimistic}.

\paragraph{Contributions.} In this paper, we first study the dynamics of MMWU in quantum zero-sum games, utilizing tools from information theory and classical game theory (Section \ref{sec:mwu}). We provide bounds on the rate of change of the total quantum relative entropy between a fully mixed Nash equilibrium and the evolving state of the system. We then study the continuous counterpart of MMWU, which we call quantum replicator dynamics (Section \ref{sec:replicator}). 
 We show that the aforementioned total quantum relative entropy is a constant of motion. Furthermore, the dynamics 
 do not converge to equilibrium (in the day-to-day sense) but rather exhibit a weaker form of approximate periodicity known as \emph{Poincar\'e recurrence}. Our proof of this result departs from the standard classical method, representing our main technical novelty. Finally, we present several simulations which corroborate our theoretical results (Section \ref{sec:experiments}).


\section{Preliminaries and Definitions}
\label{sec:prelims}

\subsection{Quantum Theory}

\paragraph{Basic concepts.}
Similarly to \cite{jain2009parallel} and in order to make the presentation self-contained, we will start by introducing some
 basic concepts of quantum information and game theory.

First, we refer to a quantum \emph{register} as a collection of qubits representing a message that is transferred from one party to another. We associate a vector space $\mathcal{H} = \mathbb{C}^n$ with any quantum register. This intuitively represents the maximum number of distinct classical states that can be stored in the register without error. The \emph{state} of a quantum register is represented by a \emph{density matrix}, which is an $n\times n$ positive semi-definite matrix with trace $1$. We will use $D(\mathcal{H})$ to denote the set of all density matrices associated with a register that is described by $\mathcal{H}$. One can naturally view such density matrices as linear operators acting on $\mathcal{H}$.

When two registers with associated spaces $\mathcal{A} = \mathbb{C}^n$ and $\mathcal{B} = \mathbb{C}^m$ are considered as a joint register, the associated space is the tensor product $\mathcal{A}\otimes \mathcal{B} = \mathbb{C}^{nm}$. If the two registers are independently prepared in states described by $\rho$ and $\sigma$, then the joint state is described by the $nm \times nm$ density matrix $\rho \otimes \sigma$. 

Next, for a given vector space $\mathcal{H}=\mathbb{C}^n$, we define $\text{L}(\mathcal{H})$ as the set of all $n\times n$ complex matrices. Furthermore, we denote the subset of $\text{L}(\mathcal{H})$ given by \emph{Hermitian} matrices as $\text{Herm}(\mathcal{H})$. A Hermitian matrix $A$ satisfies the equality $A=A^\dag$, where $A^\dag$ denotes the \emph{adjoint} (or conjugate transpose) of matrix $A$. Subsequently, we define $\text{Pos}(\mathcal{H})$ as the subset of $\text{Herm}(\mathcal{H})$ which consists of all positive semi-definite $n\times n$ matrices.

Finally, the \emph{Hilbert-Schmidt inner product} on $\text{L}(\mathcal{H})$ is defined as $\langle A, B\rangle = \Tr(A^\dag B)$
for all $A,B \in \text{L}(\mathcal{H})$. Note that $\langle A, B\rangle$ is a real number for any Hermitian matrices $A$ and $B$, and is also additionally non-negative if $A$ and $B$ are positive semi-definite.

\paragraph{Measurements and observables.} In the context of game theory, we are also interested in the concepts of quantum \emph{measurements} and \emph{observables}. An observable is simply a property of the quantum system which is measurable. The measurement of a register having associated vector space $\mathcal{H} = \mathbb{C}^n$ is a collection of linear operators $\{P_i : 1\leq i\leq k\} \subset \text{Pos}(\mathcal{H})$ which satisfies $\sum_{i=1}^k P_i=\mathbb{1}_\mathcal{H}$, where $\mathbb{1}_\mathcal{H}$ is the identity matrix on $\mathcal{H}$. If the register corresponding to $\mathcal{H}$ is in a state defined by density matrix $\rho$ and the measurement described by $P_i$ is performed, each outcome $i$ will be observed with probability $\langle P_i, \rho \rangle$. An important note is that two mixed states with the same density matrix are indistinguishable from each other by any measurement. 
\paragraph{Additional notation and definitions.} A linear mapping of the form $\Phi:\mathrm{L}(\mathcal{B}) \to \mathrm{L}(\mathcal{A})$ is called a super-operator. The adjoint super-operator to $\Phi$ is $\Phi^\dagger:\mathrm{L}(\mathcal{A}) \to \mathrm{L}(\mathcal{B})$ and is uniquely determined by the condition:
\begin{equation*}
    \langle A, \Phi(B)\rangle = \langle \Phi^\dagger(A), B\rangle 
\end{equation*}
A super-operator $\Phi:\mathrm{L}(\mathcal{B})\to\mathrm{L}(\mathcal{A})$ is said to be positive if $\Phi(P)$ is positive semi-definite for every choice of positive semi-definite operator $P \in \mathrm{Pos}(\mathcal{B})$. In addition, $\Phi^\dagger$ is positive if and only if $\Phi$ is positive.
There is a one-to-one and linear correspondence between the collection of operators of the form $R\in \mathrm{L}(\mathcal{A}\otimes\mathcal{B})$ and the collection of super-operators $\Phi:\mathrm{L}(\mathcal{B})\to\mathrm{L}(\mathcal{A})$ defined above. Specifically, for each super-operator $\Phi$ we can define an operator $R$ as follows:
\begin{equation}
    R = \sum_{1\leq i,j\leq m} \Phi(E_{i,j}) \otimes E_{i,j}
\end{equation}
where $E_{i,j}$ is the matrix with a $1$ in entry $(i,j)$ and $0$ elsewhere. Conversely, given an operator $R\in \mathrm{L}(\mathcal{A}\otimes\mathcal{B})$, we can define:
\begin{equation}\label{eqn:superoperator}
    \Phi(B) = \mathrm{Tr}_\mathcal{B} (R(\mathbb{1}_\mathcal{A}\otimes B^\top))
\end{equation}
This correspondence is linear and one can translate back and forth between the two as needed for an given application. $R$ is assumed to be positive semi-definite, so the corresponding super-operator $\Phi$ is also positive. 

\subsection{Quantum Game Theory}
A quantum system which can be manipulated by any number of agents, and where the utility of the moves is well defined, quantified and ordered can be conceived as a quantum game. In particular, a 2-player zero-sum quantum game sees players Alice and Bob each sending a quantum state to a referee, who then performs a measurement on these two states to determine their payoffs. Let $\mathcal{A}=\mathbb{C}^n$ and $\mathcal{B}=\mathbb{C}^m$ be the vector spaces that correspond to the state that Alice and Bob send to the referee. 


In order to determine the payoffs of the players' actions, the referee performs a joint measurement where Alice's and Bob's states are viewed as a single register. Thus, the referee's measurement can be described by a collection
\begin{equation}
    \{R_a : a\in\mathcal{S}\} \subset \text{Pos}(\mathcal{A}\otimes \mathcal{B})
\end{equation}
which satisfies the condition $\sum_{a=1}^k R_a = \mathbb{1}_{\mathcal{A}\otimes\mathcal{B}}$. We associate each possible measurement outcome $a$ with a payoff for each player. Since we are only considering  zero-sum games, if the payoff that Alice receives from the referee is $v(a)$, then Bob's corresponding payoff will be $-v(a)$. Henceforth, we refer to the states sent by Alice and Bob to the referee as $\rho$ and $\sigma$ respectively. For a given choice of $\rho$ and $\sigma$, Alice's expected payoff is:
\begin{equation}\label{eqn:payoffobservable}
    u(\rho, \sigma) = \sum_{a=1}^k v(a) \langle R_a, \rho\otimes \sigma \rangle = \langle R, \rho\otimes\sigma \rangle
\end{equation}
where $R = \sum_{a=1}^k v(a) R_a$. Likewise, Bob's corresponding expected payoff is $-u(\rho,\sigma) = -\langle R,\rho\otimes \sigma \rangle$. $R$ is referred to throughout the rest of the paper as a \emph{payoff observable}. A necessary and sufficient condition for matrix $R$ acting on $\mathcal{A}\otimes\mathcal{B}$ to be obtained from some real valued payoff function $v$ is that $R$ is Hermitian. From Equation \ref{eqn:superoperator}, we can equivalently define the expected payoff of Alice as $u(\rho, \sigma) = \langle \rho, \Phi(\sigma) \rangle$. We will use the latter formulation for the remainder of the paper.

A key notion of equilibrium in classical game theory is the Nash equilibrium. In the quantum setting, we define the pair of quantum states $(\rho^*, \sigma^*)$ as a Nash equilibrium of $R$ if 
\begin{equation}
    u(\rho^*, \sigma^*) \geq u(\rho, \sigma^*) \quad \text{and} \quad u(\rho^*, \sigma^*) \geq u(\rho^*, \sigma) 
\end{equation}
for all $\rho, \sigma$. That is, neither Alice nor Bob would prefer to unilaterally deviate from playing $\rho^*$ and $\sigma^*$ respectively.
On another note, since the set of available strategies (density matrices) for both agents is convex and compact and the utility function of Alice is bilinear, standard extensions of the 
of von Neumann’s Min-Max Theorem apply~\cite{fan1953minimax}.

Finally, we define the notion of a `fully mixed' Nash equilibrium. An equilibrium $(\rho^*,\sigma^*)$ is fully mixed if $\rho^*,\sigma^*$ are full rank.  

\subsection{Information Theory}

We also introduce several information theoretic concepts which will be referenced throughout the paper. 

\paragraph{Shannon entropy.} The Shannon entropy of a random variable $X$ where each strategy $x$ is obtained with probability $p(x)$ is given by $
     H(X) = -\sum_x p(x) \log{p(x)}$.
This intuitively is a measure of randomness or uncertainty in the system. A natural generalization of the Shannon entropy to a quantum context is the von Neumann entropy. For a quantum mechanical system defined by density matrix $\rho$, the von Neumann entropy is given by $S(\rho) = -\Tr (\rho \log{\rho})$.
 \paragraph{Bregman divergence.} We are also interested in the notion of Bregman divergence, which measures the distance between two points. Let $x^*\in\mathcal{X}$ be a Nash equilibrium and let $x\in\mathcal{X}$ be an arbitrary strategy profile. Also, let F be a continuously-differentiable, strictly convex function. The Bregman divergence from $x^*$ to $x$ is given by $D_{F} (x^* \| x) = F(x^*) - F(x) - \langle \nabla F(x), x^*-x \rangle$.

In the context of quantum games, the Bregman divergence is defined using matrix notation. We define the quantum relative entropy between two quantum states using the von Neumann entropy $S(\rho)$. In particular, the quantum relative entropy between two quantum states $\rho$ and $\sigma$ is defined as $S(\rho \| \sigma) = \mathrm{Tr}\left(\rho (\log \rho-\log \sigma)\right)$.

\subsection{Dynamical Systems}
\paragraph{Flows.}
Consider a differential equation $\dot{p} = f(p)$ on a topological space $\mathcal{P}$.
The existence and uniqueness theorem for ordinary differential equations guarantees that we can write the unique solution to the differential equation as a continuous map $\phi : \mathcal{P} \times \mathcal{H} \to \mathcal{P}$. This is referred to as the \emph{flow of the differential equation} such that for any point $p \in \mathcal{P}$, $\phi(p,-)$ defines a function of time corresponding to the trajectory of $p$. Conversely, fixing a time $t$ provides a map $\phi^t \equiv \phi(-,t) : \mathcal{P} \to \mathcal{P}$. In Section \ref{sec:replicator}, we introduce the notion of quantum replicator dynamics, which are Lipschitz continuous differential equations. Hence, a unique flow $\phi$ of these replicator dynamics exists.
\paragraph{Liouville's theorem.}
Liouville's theorem can be applied to any system of ordinary differential equations with a continuously differentiable vector field $\xi$ on an open domain $\mathcal{Y}\in \mathbb{R}^d$. The divergence of $\xi$ at $y\in \mathcal{Y}$ is the trace of the Jacobian at $y$: $\text{div }\xi(y) = \sum_{i=1}^d \frac{\partial \xi_i}{\partial y_i}(y)$. Because the divergence is continuous, it is integrable on Lebesgue measurable subsets of $\mathcal{Y}$. Given any such set C, define the image of C under flow $\phi$ at time $t$ as $C(t) = \{\phi(c,t) : c \in C\}$. $C(t)$ is measurable and of volume $\text{vol}[C(t)] = \int_{C(t)} d\mu$. Liouville's formula states that the time derivative of the volume $\text{vol}[C(t)]$ exists and links it to the divergence of $\xi$:
\begin{equation}
    \frac{d}{dt}\left[\text{vol }C(t)\right] = \int_{C(t)} \text{div}(\xi) d\mu
\end{equation}
If div $\xi(y)$ is null at any $y\in \mathcal{Y}$, then the volume is conserved. Since div $\xi$ is continuous, the converse statement is also true - if the volume is conserved on any open set, div $\xi(y)$ has to be null at any point $y\in \mathcal{Y}$.



\paragraph{Diffeomorphisms of flows.} A function $\mathbf{f}$ between two topological spaces is called a diffeomorphism if \begin{inparaenum}[i)]
\item $\mathbf{f}$ is a bijection,
\item $\mathbf{f}$ is continuously differentiable,
\item $\mathbf{f}$ has a continuously differentiable inverse.
\end{inparaenum}
Two flows $\Phi^t : A \to A$ and $\Psi^t : B \to B$ are diffeomorphic if there exists a diffeomorphism $\mathbf{f}:A\to B$ such that for each $x\in A$ and $t\in \mathbb{R}$, $\mathbf{f}(\Phi^t(p)) = \Psi^t(\mathbf{f}(p))$. For the purpose of our analysis, the replicator dynamics defined in Equations \ref{eqn:replicator} are translated via a diffeomorphism from the interior of $\mathcal{P}$ to a space $\mathcal{C} = \Pi_{i\in V} \mathbb{R}^{n-1}$, which allows us to show certain desirable properties.


\paragraph{Poincar\'{e} recurrence.} The concept of Poincar\'{e} recurrence arises from Henri Poincar\'{e}'s celebrated 1890 work regarding the three body problem \cite{poincare1890probleme}. He proved that if a dynamical system preserves volume and always remains bounded in its orbits, almost all trajectories return arbitrarily close to their initial position, and do so infinitely often. 

\begin{theorem}[Poincar\'{e} recurrence] \label{thm:poincareoriginal}
If a flow preserves volume and has only bounded orbits then for each open set there exist orbits that intersect the set infinitely often.
\end{theorem}


\section{MMWU in Quantum Zero-Sum Games}
\label{sec:mwu}

In \cite{jain2009parallel}, the MMWU algorithm for zero-sum games is shown to exhibit time-average convergence to an approximate Nash equilibrium in two-player quantum zero-sum games. 
The MMWU algorithm is shown in Algorithm \ref{alg:quantum}.

\begin{algorithm}[tb]
   \caption{Parallel Approximation of Equilibrium Point\label{alg:quantum}}
\begin{algorithmic}
   \STATE Let $\mu=\epsilon/8$ and let $N=\ceil{64\log{(nm)}/\epsilon^2}$.
   \STATE {\bfseries Initialize}:  $A_0 = \mathbb{1}_\mathcal{A}$, $\rho_0 = A_0/\mathrm{Tr}(A_0)$, $B_0 = \mathbb{1}_\mathcal{B}$, and $\sigma_0 = B_0/\mathrm{Tr}(B_0)$.
   \FOR{$j=1\dots N$}
   \STATE {$A_j = \exp\left(\mu \sum_{i=0}^{j-1}\Phi(\sigma_i)\right)$}
   \STATE {$\rho_j = A_j/\mathrm{Tr}(A_j)$}
   \STATE {$B_j = \exp\left(-\mu \sum_{i=0}^{j-1}\Phi^*(\rho_i)\right)$}
   \STATE {$\sigma_j = B_j/\mathrm{Tr}(B_j)$}
   \ENDFOR
\end{algorithmic}
\end{algorithm}


Note that here we focus specifically on two-player games, and utilize the expected payoffs defined via super-operators $\Phi : \mathrm{L}(\mathcal{A}) \to \mathrm{L}(\mathcal{B})$ as seen in Equation \ref{eqn:superoperator}. Moreover, $\mu$ is the step-size in the quantum algorithm.


In this section, we examine closely the update steps for each player in the MMWU algorithm (Algorithm \ref{alg:quantum}) and analyze the limiting behaviour of the total quantum relative entropy in the system. First, we introduce two useful facts which will aid in the analysis.

\begin{fact}[Golden-Thompson inequality \cite{goldeninequality, thompsoninequality}]\label{fact:golden}
Let $A,B$ be Hermitian matrices. Then
\begin{equation}
    \Tr \exp(A+B) \le \Tr\exp(A)\exp(B)
\end{equation}
\end{fact}
\begin{fact}\label{fact:exp}
Let $0 \leq A \leq \id$ be a PSD matrix and $\delta$ be a real number. Then,
\begin{equation}
    \exp(\delta A) \le \id + \delta \exp(\delta) A
\end{equation} 

\end{fact}

Next, we put forward two corollaries which will help us prove Theorem~\ref{thm:qre_mwu}. 
 We first define $\Delta S(\rho^* \| \rho_j) = S(\rho^* \| \rho_j) - S(\rho^*\| \rho_{j-1})$ and $\Delta S(\sigma^* \| \sigma_j) = S(\sigma^* \| \sigma_j) - S(\sigma^*\| \sigma_{j-1})$.


\begin{corollary}\label{corollary:equality}
    The change in the sum of quantum relative entropies in a quantum zero-sum game between a fully-mixed Nash equilibrium and the players' strategies is given by:
\begin{align}
    \Delta S(\rho^* \| \rho_j) + \Delta S(\sigma^* \| \sigma_j)= \log{\frac{\mathrm{Tr}A_j}{\mathrm{Tr}A_{j-1}}}+\log{\frac{\mathrm{Tr}B_j}{\mathrm{Tr}B_{j-1}}}
\end{align}
\end{corollary}

\begin{corollary}\label{corollary:inequality}
    The following trace inequalities hold for PSD matrices $A$ and $B$ updated with MMWU:
    
    \begin{align}
        \Delta S(\rho^* \| \rho_j) + \Delta S(\sigma^* \| \sigma_j) &\geq \mu\exp(-\mu) \Tr \left(\rho_{j} \Phi(\sigma_{j-1})\right) - \mu\exp(\mu) \Tr \left(\rho_{j-1} \Phi(\sigma_{j})\right)\label{eqn:lowerboundsum}\\
        \Delta S(\rho^* \| \rho_j) + \Delta S(\sigma^* \| \sigma_j) &\leq \mu\exp(\mu) \Tr \left(\rho_{j-1} \Phi(\sigma_{j-1}) \right) -\mu\exp(-\mu) \Tr \left(\rho_{j-1} \Phi(\sigma_{j-1})\right) \label{eqn:upperboundsum}
    \end{align}
\end{corollary}
The proof of Corollary~\ref{corollary:inequality} relies on Facts \ref{fact:golden} and \ref{fact:exp}, as well as the equality shown in Corollary~\ref{corollary:equality}.

In many practical scenarios, one would use decreasing step-sizes when running MMWU. As such, we utilize Corollary~\ref{corollary:inequality} and take the limit as step-size $\mu$ goes to $0$ in order to show the following theorem:

\begin{theorem}\label{thm:qre_mwu}
    The change in the sum of quantum relative entropies between a fully mixed Nash equilibrium and the player's strategies and in a two-player zero-sum quantum game tends to zero when step-size $\mu\to 0$. Specifically,
    \begin{equation}
        \lim_{\mu \to 0}\frac{1}{\mu} \left(\Delta S(\rho^* \| \rho_j) + \Delta S(\sigma^* \| \sigma_j)\right) = 0 
    \end{equation}
\end{theorem}
The proof of Theorem~\ref{thm:qre_mwu}, along with Corollaries \ref{corollary:equality} and \ref{corollary:inequality} can be found in Appendix \ref{appsecs:proofs_mwu}.

Theorem \ref{thm:qre_mwu} further motivates an investigation into the continuous time variant of MMWU, which we call quantum/matrix replicator dynamics. In particular, we show that in the continuous case, the sum of quantum relative entropies is invariant. 
We introduce and study quantum replicator dynamics in detail in Section \ref{sec:replicator}.


\section{Replicator Dynamics in Quantum Zero-Sum Games}
\label{sec:replicator}
We have seen that in the case of discrete dynamics (MMWU), as the step-size becomes infinitesimal, the total quantum relative entropy between the Nash equilibirum and the system state tends to  stabilize.
A natural question would then be: does the same result hold in continuous time?  In order to explore this question in greater detail, we first need to define the continuous analogue of MMWU, the \emph{quantum replicator dynamics}, which has well known classical analogues \cite{sandholm2010population}.

We start by rewriting the MMWU update steps, but now defined over a continuous time interval $[0,t]$:

\begin{eqnarray}
A(t) &=& \int_0^t \Phi(\sigma(\tau))d\tau \label{eqn:A_t}\\ 
\rho(t) &=& \exp(A(t))/\mathrm{Tr}(\exp(A(t))) \label{eqn:rho_t}\\
B(t) &=& - \int_0^t \Phi^\dagger(\rho(\tau))d\tau \label{eqn:B_t}\\
 \sigma(t) &=& \exp(B(t))/\mathrm{Tr}(\exp(B(t))) \label{eqn:sigma_t}
\end{eqnarray}

Note here that we shift the exponential terms from the definition of $A(t)$ and $B(t)$ to the corresponding $\rho(t)$ and $\sigma(t)$ terms. This will help simplify some of the proof techniques later on in the paper. Furthermore, in the rest of the paper we will typically drop from the notation the explicit dependence on $t$ to ease with the notational burden.

It is important to note the following observation, which will be helpful in our later analysis.
\begin{observation}\label{obs:equivalence}
The discrete-time trajectories $\rho_j$ and $\sigma_j$ defined in Algorithm \ref{alg:quantum} are a standard Euler discretization (with step $\mu$) of the continuous-time trajectories $\rho(t)$ and $\sigma(t)$ defined in Equations \ref{eqn:rho_t} and  \ref{eqn:sigma_t}.
\end{observation}

We now define the \emph{quantum replicator dynamics} as:
\begin{equation}
d\rho/dt = \frac{d}{dt}\left(\frac{\exp(A)}{\mathrm{Tr}(\exp(A))}\right),\quad
d\sigma/dt = \frac{d}{dt}\left(\frac{\exp(B)}{\mathrm{Tr}(\exp(B))}\right) \label{eqn:replicator}
\end{equation}

It is worth noting that in the classical/commuting setting, one can write the replicator equations in a form that describes the relative utility that one agent obtains as compared to the average utility overall. However, in the quantum case this is not possible in general, since it relies on the assumption that $\int_0^t \Phi(\sigma(\tau))d\tau$ and $\Phi(\sigma(t))$, and respectively $\int_0^t \Phi^\dagger(\rho(\tau))d\tau$ and $\Phi^\dagger(\rho(t))$ commute. 

As a consequence of Observation \ref{obs:equivalence}, we can also conclude that the dynamical system described by the replicator dynamics defined in Equations \ref{eqn:replicator} is a limit case of the dynamical system described by the MMWU algorithm as $\mu \to 0$.

We are now able to state the main theorem for quantum relative entropy in quantum replicator dynamics.
\begin{theorem} \label{thm:qre_replicator}
When applying matrix/quantum replicator dynamics in a quantum zero-sum game with a fully-mixed Nash equilibrium $(\rho^*,\sigma^*)$,
the sum of quantum relative entropies
between the fully-mixed Nash equilibrium and the state of the system $(\rho(t),\sigma(t))$ is invariant on every system trajectory, i.e.:
\begin{equation}
\frac{d\big(S(\rho^* \| \rho(t))+S(\sigma^* \| \sigma(t))\big)}{dt} = 0
\end{equation}
\end{theorem}

\subsection{Poincar\'e Recurrence in Quantum Zero-Sum Games}

Now that we have described analytical results surrounding the day-to-day behaviour of quantum replicator dynamics, we seek to understand the \emph{dynamics} of the trajectories. After all, invariance of quantum relative entropy does not fully describe how the system moves over time. We show that for any two-player zero-sum quantum game updated with replicator dynamics, the system exhibits \emph{Poincar\'e recurrence}, insofar as the game is zero-sum and has a fully-mixed Nash equilibria. As introduced in Section \ref{sec:prelims}, the notion of Poincar\'e recurrence is a weaker version of periodicity. To be precise, for almost all initial conditions $\rho_0 \in \mathcal{P}$, the replicator dynamics return arbitrarily close to $\rho_0$ infinitely often. 

\begin{theorem} \label{thm:recurrence}
The quantum replicator dynamics given in Equations \ref{eqn:replicator} are Poincar\'e recurrent in any two player zero-sum game which has a fully-mixed Nash equilibrium.
\end{theorem}

The proof of this main theorem involves carefully piecing together several auxiliary results, which we will describe in the rest of the section. Furthermore, we stress that due to the non-commutative nature of quantum systems, the standard (classical) approach of differentiating the discrete-time dynamics in the primal space of probability distributions does not apply directly unless we have the highly unlikely situation where $\int_0^t \Phi(\sigma(\tau))d\tau$ and $\Phi(\sigma(t))$ (resp. $\int_0^t \Phi^\dagger(\rho(\tau))d\tau$ and $\Phi^\dagger(\rho(t))$) commute. This problem with carrying over the standard approach is explicitly discussed in Appendix \ref{appsecs:proofs_replicator}. 

For the proof in the quantum setting, we first define a \emph{canonical transformation} on the space of the matrices $A(t)$ and $B(t)$, which will be crucial in proving the theorem. 

\begin{definition}[Canonical transformation]\label{def:canonical}
    We define the canonical transformation of $A'(t)$ and $B'(t)$ to be a mapping of $A(t)$ and $B(t)$ as defined by Equations \ref{eqn:A_t} and \ref{eqn:B_t}. In particular, we define
    \begin{equation}\label{eqn:canonical}
        \begin{aligned}
        A'(t) &{}= A(t) - (v^\dag A(t) v)\mathbb{1}\\
        B'(t) &{}= B(t) - (v^\dag B(t) v)\mathbb{1}
    \end{aligned}
    \end{equation}
    where $v$ is a fixed vector defined as $v = [1, 0 \dots 0]^\top$, such that the values of $v^\dag A(t) v$ and $v^\dag B(t) v$ are real numbers corresponding to the $(1,1)$-th element of matrices $A(t)$ and $B(t)$ for all $t$.
    Notice that this creates matrices $A'(t)$ and $B'(t)$ which have $0$ as the $(1,1)$-th entry.
\end{definition}
Under the transformation in Definition~\ref{def:canonical}, the vector fields $\dot{A}'(t) = F(A')$ and $\dot{B}'(t) = F(B')$ are given by:
\begin{equation}\label{eqn:vectorfield}
\begin{aligned}
    \dot{A}'(t) &{}= \Phi(\sigma(t)) - (v^\dag\Phi(\sigma(t))v)\mathbb{1}\\
    \dot{B}'(t) &{}= -\Phi^\dagger(\rho(t)) + (v^\dag\Phi^\dagger(\rho(t))v)\mathbb{1}
\end{aligned}
\end{equation}
where $\frac{d}{dt}\left(v^\dag A(t) v\right)$ is given by $v^\dag \frac{dA(t)}{dt} v$.

Moreover, the values of $\rho'(t)$ and $\sigma'(t)$ are defined as:
\begin{equation}
\begin{aligned}
    \rho'(t) &{}= \exp(A'(t))/\mathrm{Tr}(\exp(A'(t)))\\
    \sigma'(t) &{}= \exp(B'(t))/\mathrm{Tr}(\exp(B'(t)))
\end{aligned}
\end{equation}


\begin{proposition}\label{prop:map1}
    The dynamics of $\rho(t)$ and $\sigma(t)$ remain the same after undergoing the canonical transformation. Equivalently, $A'(t)$ and $A(t)$ $($resp. $B'(t)$ and $B(t))$ admit the same strategy $\rho(t)$ $($resp. $\sigma(t))$.
\end{proposition}

\begin{proposition}\label{prop:diffeomorphic}
    The mappings $A'(t)$ and $\rho(t)$ $($resp. $B'(t)$ and $\sigma(t))$ are diffeomorphic to one another.
\end{proposition}
Proposition \ref{prop:diffeomorphic} will be of crucial importance to our proof technique, since we first prove recurrence for the system described by $A'(t)$ and $B'(t)$, then recover recurrence in $\rho(t)$ and $\sigma(t)$. 

To show Poincar\'e recurrence of Equations~\ref{eqn:replicator}, we first prove two key properties: (i) the flow of $\dot{A'}$ is volume preserving, meaning that the trace of the Jacobian of the respective vector fields $\dot{A'}(t) = F(A')$  and $\dot{B'}(t) = F(B')$ are zero, and (ii) $A'$ and $B'$ have bounded orbits from any interior initial condition. Then, Poincar\'e recurrence of $A'$ and $B'$ follows from Poincar\'{e}'s recurrence theorem.

\paragraph{Volume Conservation.} We introduce a lemma which shows that in two-player zero-sum quantum replicator dynamics, the canonical transformation produces a dynamical system which preserves volume.

\begin{lemma} \label{lemma:volumeconservation}

For two-player zero-sum quantum replicator dynamics, the vector fields that arise as a result of the canonical transformation in Definition~\ref{def:canonical} are volume preserving.
\end{lemma}
The proof of Lemma \ref{lemma:volumeconservation} follows from considering the flows of $A'(t)$ and $B'(t)$ and showing that the divergences of the vector fields defined in Equations~\ref{eqn:vectorfield} are equal to zero. A straightforward application of Liouville's theorem completes the proof.

\paragraph{Bounded Orbits.} We now show that the transformed dynamical system always has bounded orbits when initialized on the interior of the space of probability density matrices.
\begin{lemma} \label{lemma:boundedorbits}
    For any finite initial points $A(0)$ and $B(0)$, the dynamics mapped to $A'(t)$ and $B'(t)$ via the transformation in Definition~\ref{def:canonical} have bounded orbits.
\end{lemma}
The proof of Lemma \ref{lemma:boundedorbits} relies on Theorem~\ref{thm:qre_replicator}, and leverages the hermicity of the matrices involved. Moreover, we use the fact that the canonically transformed matrices $A'(t)$ and $B'(t)$ have zero as the $(1,1)$-th element to bound the eigenvalues of $A'(t)$ and $B'(t)$ away from infinity.

\noindent Now we are ready to prove Theorem \ref{thm:recurrence} using Lemmas \ref{lemma:volumeconservation} and \ref{lemma:boundedorbits}.
\begin{proof}[Proof of Theorem \ref{thm:recurrence}]
    By Lemmas \ref{lemma:volumeconservation} and \ref{lemma:boundedorbits}, as well as the Poincar\'e recurrence theorem introduced in Section \ref{sec:prelims}, we immediately see that the system of replicator equations given by $dA'/dt$ and $dB'/dt$ are Poincar\'e recurrent since they are volume preserving and have bounded orbits. 
    Since the flows of $A'(t)$ and $\rho(t)$ are diffeomorphic to one another (likewise for $B'(t)$ and $\sigma(t)$), $d\rho/dt$ and $d\sigma/dt$ are also Poincar\'e recurrent. This concludes the proof.
\end{proof}
All proofs of the results in this section are provided in Appendix \ref{appsecs:proofs_replicator}.

\section{Experimental Results}
\label{sec:experiments}
\begin{figure}[ht]
    \centering
    \begin{minipage}{.245\linewidth}
      \centering
      \includegraphics[width=.95\linewidth]{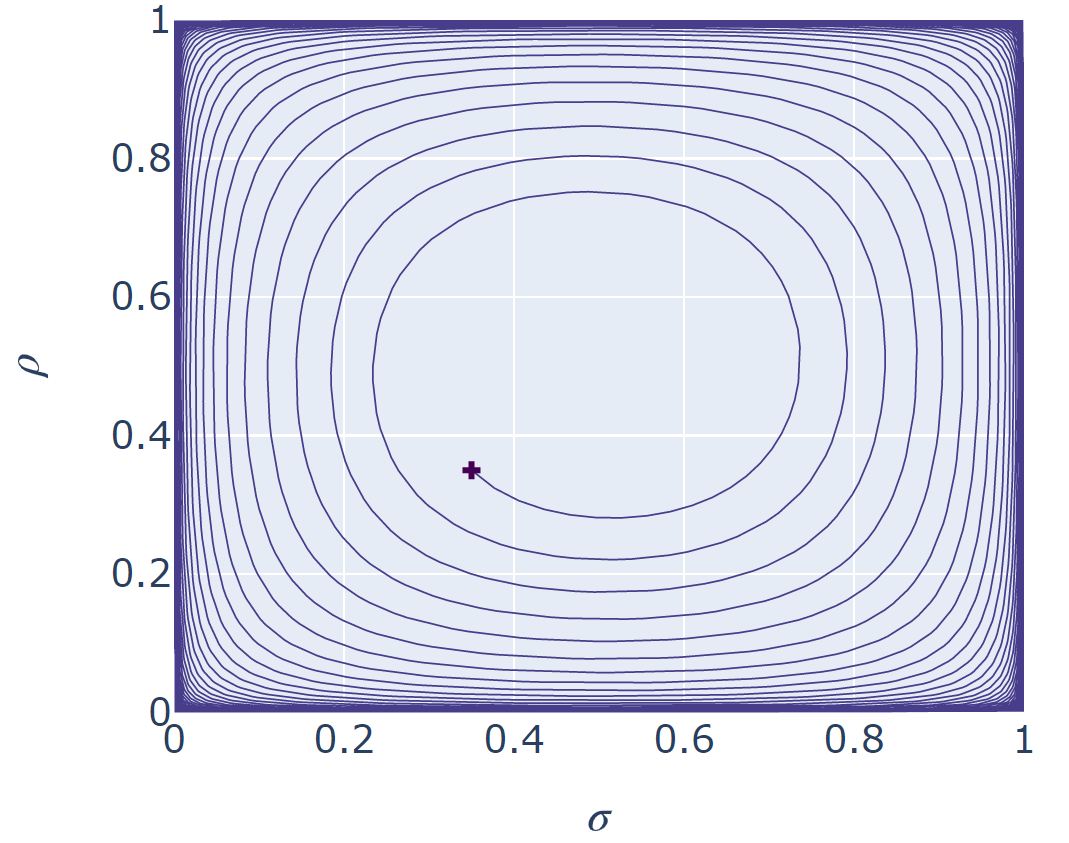}
      $\mu = \log(1+ \frac{1}{t^{1/4}})$
    \end{minipage}
    \begin{minipage}{.245\linewidth}
      \centering
      \includegraphics[width=.95\linewidth]{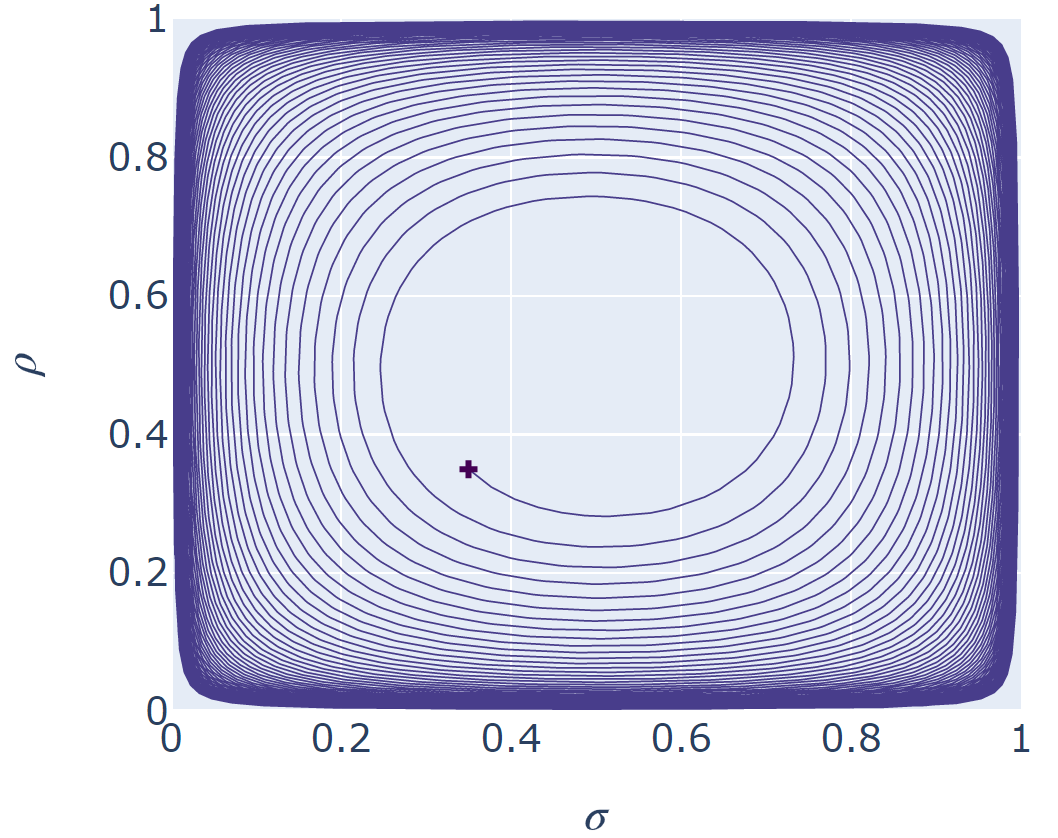}
      $\mu = \log(1+ \frac{1}{t^{1/3}})$
    \end{minipage}
    \begin{minipage}{.245\linewidth}
      \centering
      \includegraphics[width=.95\linewidth]{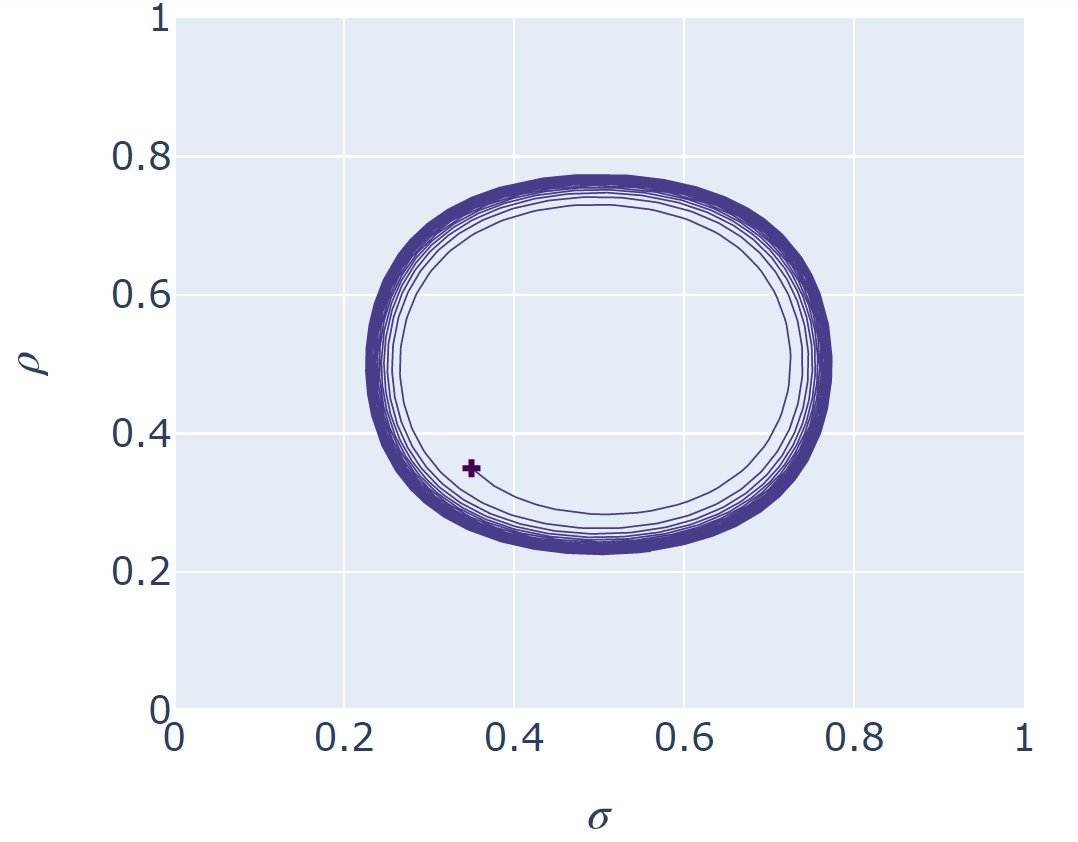}
      $\mu = \log(1+ \frac{1}{t^{1/2}})$
    \end{minipage}
    \begin{minipage}{.245\linewidth}
      \centering
      \includegraphics[width=.95\linewidth]{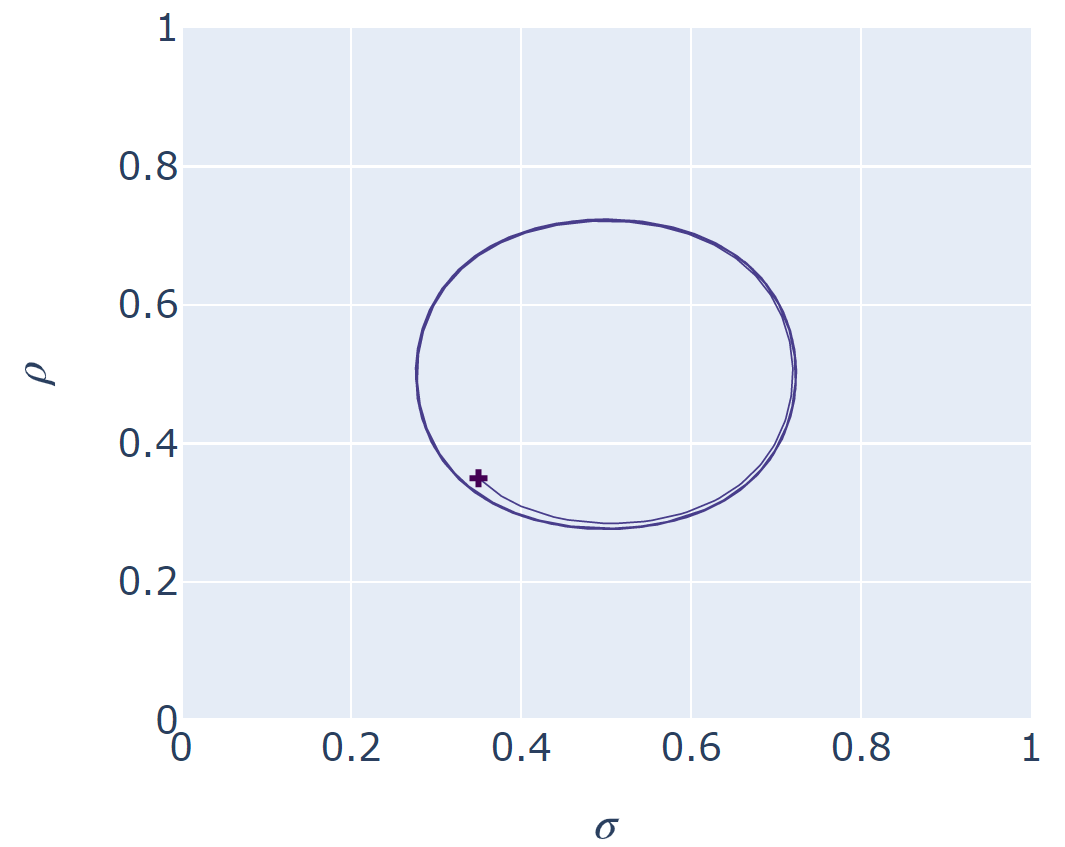}
      $\mu = \log(1+ \frac{1}{t^{2/3}})$
    \end{minipage}%
    \caption{Eigenvalue trajectories for quantum Matching Pennies game with decreasing $\mu$ values.}
    \label{fig:quantum_rps_entropy}
\end{figure}

To corroborate the theoretical results presented in prior sections, we performed relevant simulations of quantum games using both discrete MMWU and replicator dynamics. In the rest of this section, we standardize the use of quantum game matrices obtained via basis transform (described in more detail in Appendix \ref{appsecs:experiments}). This effectively allows us to transform classical games to the matrix setting.


First, we show the  trajectories of the first eigenvalue of each player in a quantum Matching Pennies game, obtained using the discrete MMWU algorithm. We see that in accordance to Theorem \ref{thm:qre_mwu}, the rate of divergence of the trajectories from the uniform Nash goes to zero for cases with rapidly decreasing learning rate $\mu$. 

In the case of replicator dynamics, we present Bloch sphere representations of the trajectories in a quantum Matching Pennies game. The Bloch sphere is a unit 2-sphere representation of a qubit, and we utilize it to visualize the orbits of the replicator dynamics. In particular, the density matrix representing the strategy of each player at each time-step is given as a point within the sphere, and we plot the movement of these orbits over time. According to Theorems \ref{thm:qre_replicator} and \ref{thm:recurrence}, we expect the trajectories of the replicator dynamics to stay on the interior of the Bloch sphere, since the surface of the sphere corresponds to the pure states of the system. We see from Figure \ref{fig:bloch_sphere} that over time, the system never reaches the boundary of the sphere, which experimentally agrees with our theory.

\begin{figure}[ht]
    \centering
    \begin{minipage}{.245\linewidth}
      \centering
      \includegraphics[width=0.95\linewidth]{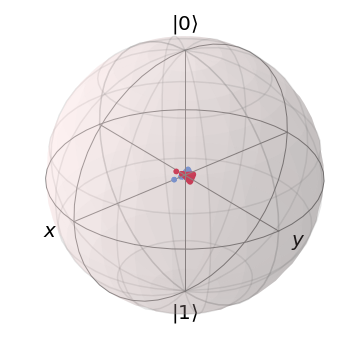}
      $t=1$
    \end{minipage}%
    \begin{minipage}{.245\linewidth}
      \centering
      \includegraphics[width=0.95\linewidth]{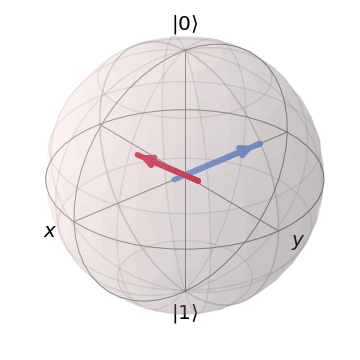}
      $t=50$
    \end{minipage}
    \begin{minipage}{.245\linewidth}
      \centering
      \includegraphics[width=0.95\linewidth]{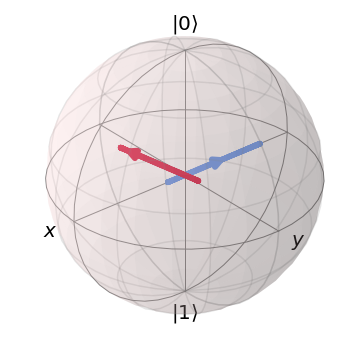}
      $t=370$
    \end{minipage}
    \begin{minipage}{.245\linewidth}
      \centering
      \includegraphics[width=0.95\linewidth]{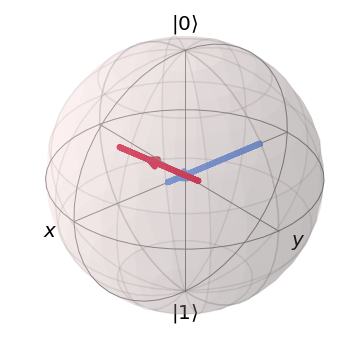}
      $t=1000$
    \end{minipage}
    \caption{Bloch sphere trajectories for quantum Matching Pennies game between Alice (blue) and Bob (red). The arrowheads represent the current state of each player at each time-step. Notice that over time, the orbits oscillate within the interior of the Bloch sphere.}
    \label{fig:bloch_sphere}
\end{figure}

Finally, in the experiments above we have only studied single qubit systems. We provide further experiments which show that our results hold beyond the single qubit setting in Appendix \ref{appsecs:largerscale}.

\section{Conclusion}
\label{sec:conclusion}
In this paper, we studied the properties of  Matrix Multiplicative Weights Update and its continuous analogue, quantum replicator dynamics, in the context of two-player zero-sum quantum games. First, we provide a formulation of quantum replicator dynamics which arises from MMWU. Then, we show that such systems exhibit quantum information-theoretic constant of motions. Finally, we show that in quantum replicator systems with interior Nash equilibria, the dynamics exhibit Poincar\'e recurrence. 

This work constitutes an initial step towards analyzing learning behaviour in games with quantum information. In the classical world, showing that conservation laws and recurrence holds has led to a better understanding of game dynamics in increasingly complex settings \cite{skoulakis2020evolutionary, nagarajan2020chaos, perolat2020poincare}. Similar work is now potentially possible in the quantum setting.

Moreover, one key implication which can be derived from our recurrence results is to encourage novel discretization methods that preserve the connections to volume preservation and more generally, conservative dynamical systems (e.g. Hamiltonians). This line of research has received increased interest in recent years \cite{wibisono2016variational,wibisono2022alternating, diakonikolas2021generalized}. It is a natural question to explore whether such techniques can lead to similar advantages in the quantum setting as well.

Finally, some other interesting directions for future work include: \begin{itemize}
    \item extending our results to multi-agent network generalizations of quantum zero-sum games,
    \item understanding the potential behavior of quantum GANs using our results about learning in quantum zero-sum games, and
    \item understanding the quantum setting for different classes of games, e.g., potential games.
\end{itemize}

\bibliographystyle{apalike}
\bibliography{ref}

\newpage
\begin{appendices}
This supplementary material contains an overview of additional related work in Appendix~\ref{appsecs:related}, 
additional examples and definitions in Appendix~\ref{appsecs:definitions}, proofs omitted from the paper for space considerations in Appendix~\ref{appsecs:proofs_mwu} and~\ref{appsecs:proofs_replicator}, and further experimental results and details in Appendix~\ref{appsecs:experiments}. 

\section{Additional Related Work}
\label{appsecs:related}
We now describe a broader class of related work. In particular, we focus on work regarding $a)$ the MMWU algorithm and its applications as well as $b)$ quantum algorithms in game theory.

\paragraph{MMWU Algorithm and Applications.} MMWU as an algorithm has been widely studied in the past. For instance, MMWU has been applied to the problem of solving semidefinite programs in \cite{brandao2017quantum}. The central idea behind this line of work is the `quantization' of standard SDP solvers using the matrix multiplicative weights method. Furthermore, MMWU has also been applied to the problem of learning a quantum state, also known as state tomography. \cite{aaronson2019online, youssry2019efficient} have utilized MMWU as a framework for online learning in this setting. 

In a broader sense, MMWU can also be considered a member of the class of Matrix Follow-The-Regularized-Leader (FTRL) algorithms. This paradigm has been studied in works such as \cite{hazan2012near, allen2015spectral, hopkins2020robust}, which admits a wider range of applications such as spectral sparsification of graphs. In addition, some more recent works such as \cite{allen2017follow} study extensions to FTRL within the regret minimization framework.

\paragraph{Quantum Computing.} Quantum information theory has proved to be a useful tool in the fields of reinforcement learning and deep learning \cite{bondesan2021hintons, Daoyi_Dong_2008, dunjko2017advances}. The basic results of classical game theory, such as von Neumann's minimax theorem and more have been shown to hold in the quantum context \cite{osti_1513429, Accardi_2020}, while recent work has shown analogous complexity and equilibrium results between classical and quantum games \cite{zhang2012quantum, bostanci2021quantum}.

There is also a rich literature towards solving and analyzing games using quantum algorithms. \cite{van2019quantum} derives a sublinear time quantum algorithm for solving two-player zero-sum games, while \cite{li2020sublinear} derives a similar algorithm for general matrix games. Finally, quantum algorithms have proven to be useful in many machine learning problems such as classification \cite{li2019sublinear} and GANs \cite{lloyd2018quantum}. For a general overview of quantum algorithms in machine learning see \cite{lloyd2013quantum} and \cite{biamonte2017quantum}.

\section{Additional Preliminaries}
\label{appsecs:definitions}

\subsection{An Example: Quantum Matching Pennies}



To provide a clearer picture of quantum games, we illustrate the differences between the classical and quantum versions of the famous Matching Pennies game. In classical game theory, Matching Pennies is defined by the payoff matrix $P =  \begin{bsmallmatrix}1 & -1\\ -1 & 1\end{bsmallmatrix}$. In the terminology defined above, the payoff observable $R$ would then simply be a diagonal $4\times4$ matrix with the elements of the classical payoff matrix on the diagonal:
\[R=
\begin{bmatrix}
1 & 0 & 0 & 0\\
0 & -1 & 0 & 0\\
0 & 0 & -1 & 0\\
0 & 0 & 0 & 1
\end{bmatrix}
\]\\
Let us first consider the case where Alice and Bob play randomized mixed strategies in a classical fashion. Alice and Bob will then simultaneously select \emph{vectors} which represent a probability distribution over the strategies \emph{heads} or \emph{tails}. For example, Alice could return the vector $x_A =  [1/2, 1/2]$ and Bob could return the vector $x_B = [1/3, 2/3]$. The expected utility for each player is then determined using the payoff matrix: 
\[ u_A =  \begin{bmatrix}
1/2 \\
1/2 
\end{bmatrix}^\top \begin{bmatrix}
1 & -1\\
-1 & 1
\end{bmatrix} \begin{bmatrix}
1/3 \\
2/3 
\end{bmatrix} = 0, 
\quad 
u_B =  \begin{bmatrix}
1/3 \\
2/3 
\end{bmatrix}^\top \begin{bmatrix}
1 & -1\\
-1 & 1
\end{bmatrix} \begin{bmatrix}
1/2 \\
1/2
\end{bmatrix} = 0
\]

For the quantum case, Alice and Bob will instead send density matrices that represent their probabilities of playing a given strategy to a referee. For example, Alice might send the matrix $x_A =  \begin{bmatrix}1/2 & 1/2i\\ -1/2i & 1/2\end{bmatrix}$ and Bob could return the matrix $x_B = \begin{bmatrix}1/3 & 1/2i\\ -1/2i & 2/3\end{bmatrix}$. The payoffs to each player are then calculated using the payoff observable R as defined in Equation \ref{eqn:payoffobservable}.

\subsection{Classical MWU Algorithm}
In the paper we focus on the matrix version of MWU, with particular focus on the quantum interpretation of the positive semi-definite matrices used in Algorithm \ref{alg:quantum}. However, many of the results we show are analogous to results shown in classical settings. 

For two-player zero-sum games (i.e. $A = -B^\top$, where $A$ and $B$ are the payoff matrices for players $a$ and $b$.), the classical MWU algorithm is shown in Algorithm \ref{alg:classical} \cite{freund1999adaptive}:
  
  
  
\begin{algorithm}[!htb]
   \caption{MWU for Two-Player Zero-Sum Games\label{alg:classical}}
\begin{algorithmic}
   \STATE {\bfseries Initialize:} $x_{i, a}^0 = 1/n$ and $x_{i, b}^0 = 1/n$ for all $i$ (both players, uniform)
   \FOR{t=1\dots T}
   \STATE Players $a$ and $b$ choose action $i\in\mathcal{S}$ with probability $x_{i,a}^t$ and $x_{i,b}^t$ respectively.
   \FOR{each action $i$}
   \STATE $x_{i, a}^t = \frac{x_{i,a}^{t-1}(1+\epsilon)^{(Ax_b^{t-1})_{i}}}{\sum_{j\in\mathcal{S}}x_{j,a}^{t-1}(1+\epsilon)^{(Ax_b^{t-1})_{j}}}$\\
    $x_{i, b}^t = \frac{x_{i,b}^{t-1}(1+\epsilon)^{(Bx_a^{t-1})_{i}}}{\sum_{j\in\mathcal{S}}x_{j,b}^{t-1}(1+\epsilon)^{(Bx_a^{t-1})_{j}}}$
 \ENDFOR
 \ENDFOR
 
\end{algorithmic}
\end{algorithm}

  
  
  
  Furthermore $(Ax_b^{t-1} )_{i}$ and $(Bx_a^{t-1} )_{i}$ represent the expected utility obtained by players $a$ and $b$ respectively for playing strategy $i$. 

\section{Proofs from Section \ref{sec:mwu}}
\label{appsecs:proofs_mwu}
In this section we present proofs of results in Section \ref{sec:mwu} which were omitted from the main paper due to space considerations.

\subsection{Proof of Corollary \ref{corollary:equality}}
\begin{proof}
We use the definition of quantum relative entropy. First, we consider the change in quantum relative entropy from timestep $j-1$ to $j$. 
    \begin{align*}
        S (\rho^* \| \rho_j) - S (\rho^* \| \rho_{j-1}) &= \mathrm{Tr}\left(\rho^* (\log \rho_{j-1}-\log \rho_j)\right)\\
        &= \mathrm{Tr}\left(\rho^* \left(\log \frac{A_{j-1}}{\mathrm{Tr}A_{j-1}}-\log \frac{A_{j}}{\mathrm{Tr}A_{j}}\right)\right)\\
        &= \mathrm{Tr}\left( \rho^*\left(\log{A_{j-1}}-\log{A_j} + \log{\mathrm{Tr}A_j}-\log{\mathrm{Tr}A_{j-1}}\right)\right)\\
        &= \mathrm{Tr}\left(\rho^*\left(-\mu \Phi(\sigma_{j-1}) + \log{\mathrm{Tr}A_j}-\log{\mathrm{Tr}A_{j-1}}\right) \right)
    \end{align*}
    Hence,
    \begin{equation}\label{eqn:kldivA}
       S (\rho^* \| \rho_j) - S (\rho^* \| \rho_{j-1}) = -\mu \mathrm{Tr} (\rho^* \Phi(\sigma^*) )+ \mathrm{Tr}( \rho^*(\log{\mathrm{Tr}A_j}-\log{\mathrm{Tr}A_{j-1}}))
    \end{equation}
Similarly, we know that 
\begin{equation}\label{eqn:kldivB}
    S (\sigma^* \| \sigma_j) - S (\sigma^* \| \sigma_{j-1}) = -\mu \mathrm{Tr}( \sigma^* \Phi^\dagger(\rho^*) )+ \mathrm{Tr}( \sigma^*(\log{\mathrm{Tr}B_{j}}-\log{\mathrm{Tr}B_{j-1}}))
\end{equation}
Summing up equations \ref{eqn:kldivA} and \ref{eqn:kldivB}:
\begin{align*}
    & \Delta S(\rho^* \| \rho_j) + \Delta S(\sigma^* \| \sigma_j)\\ 
    &= \mu\left[-\mathrm{Tr}(\rho^*\Phi(\sigma^*))+\mathrm{Tr}(\sigma^*\Phi^\dagger(\rho^*))\right] + \log{\frac{\mathrm{Tr}A_j}{\mathrm{Tr}A_{j-1}}}+\log{\frac{\mathrm{Tr}B_j}{\mathrm{Tr}B_{j-1}}} \\
    &= \log{\frac{\mathrm{Tr}A_j}{\mathrm{Tr}A_{j-1}}}+\log{\frac{\mathrm{Tr}B_j}{\mathrm{Tr}B_{j-1}}} 
\end{align*}
\end{proof}

\subsection{Proof of Corollary \ref{corollary:inequality}}
\begin{proof}
We first show the upper bound on the sum $\log\left(\frac{\Tr A_j}{\Tr A_{j-1}}\right)+ \log\left(\frac{\Tr B_j}{\Tr B_{j-1}}\right)$. Note that the definition of $A_j$ is:
$$A_j = \exp\left(\mu \sum_{i=0}^{j-1} \Phi(\sigma_i)\right)$$
We can alternatively write:
\begin{align}\label{eqn:a_j}
    A_j = \exp\left(\log(A_{j-1}) + \mu \Phi(\sigma_{j-1})\right)
\end{align}
Hence, by taking the trace of $A_{j-1}$ and applying Facts \ref{fact:golden} and \ref{fact:exp}:
\begin{align*}
\Tr (A_{j-1})  &\leq \Tr \left[ A_{j} \exp(-\mu \Phi(\sigma_{j-1}) ) \right]& \mbox{(Fact~\ref{fact:golden})} \\
&\leq \Tr \left[A_{j} (\id - \mu \exp(-\mu) \Phi(\sigma_{j-1}) )\right] & \mbox{(Fact~\ref{fact:exp})}\\
&= (\Tr (A_{j})) (1- \mu\exp(-\mu) \Tr \left(\rho_{j} \Phi(\sigma_{j-1}) \right) \\
&\leq  (\Tr (A_{j})) \exp (- \mu\exp(-\mu) \Tr \left(\rho_{j} \Phi(\sigma_{j-1}) )\right). \hfill & (\mbox{$1 + x \leq \exp(x)$})
\end{align*}
Similarly,
\begin{align*}
\Tr  (B_{j-1})  &\leq \Tr \left[ B_{j} \exp(\mu \Phi^\dagger(\rho_{j-1}) )\right] \\
&\leq \Tr \left[ B_{j} (\id + \mu \exp(\mu) \Phi^\dagger(\rho_{j-1}) )\right]\\
&= (\Tr (B_{j})) (1+ \mu\exp(\mu) \Tr \left(\rho_{j-1} \Phi(\sigma_{j}) )\right)\\
&\leq (\Tr (B_{j})) \exp (\mu\exp(\mu) \Tr \left(\rho_{j-1} \Phi(\sigma_{j}) )\right). 
\end{align*}

\noindent By rearranging and taking matrix logarithm on both sides, we obtain:
\begin{align}
        \log\left(\frac{\Tr A_j}{\Tr A_{j-1}}\right) &\geq \mu\exp(-\mu) \Tr \left(\rho_{j} \Phi(\sigma_{j-1}) \right) \ \   &\ \   \log\left(\frac{\Tr B_j}{\Tr B_{j-1}}\right) &\geq -\mu\exp(\mu) \Tr \left(\rho_{j-1} \Phi(\sigma_{j})  \right)
      \end{align}
The inequality in Equation \ref{eqn:lowerboundsum} follows by summing up the two inequalities obtained above and applying Corollary \ref{corollary:equality}. Likewise, we can perform similar calculations on $\Tr(A_j)$ and $\Tr(B_j)$ to obtain the statement of Equation \ref{eqn:upperboundsum}.
\begin{align*}
\Tr A_{j}  &\leq \Tr \left[A_{j-1} \exp(\mu \Phi(\sigma_{j-1}) )\right]\\
&\leq \Tr \left[A_{j-1} (\id + \mu \exp(\mu) \Phi(\sigma_{j-1}) )\right]\\
&= (\Tr A_{j-1}) (1 + \mu\exp(\mu) \Tr \left(\rho_{j-1} \Phi(\sigma_{j-1}) )\right) \\
&\leq  (\Tr A_{j-1}) \exp (\mu\exp(\mu) \Tr \left(\rho_{j-1} \Phi(\sigma_{j-1}) )\right)
\end{align*}
Similarly,
\begin{align*}
\Tr B_{j}  &\leq \Tr \left[B_{j-1} \exp(-\mu \Phi^\dagger(\rho_{j-1}) )\right]\\
&\leq \Tr \left[B_{j-1} (\id -\mu \exp(-\mu) \Phi^\dagger(\rho_{j-1}) )\right]\\
&= (\Tr B_{j-1}) (1 - \mu\exp(-\mu) \Tr \left(\rho_{j-1} \Phi(\sigma_{j-1}) )\right) \\
&\leq  (\Tr B_{j-1}) \exp (-\mu\exp(-\mu) \Tr \left(\rho_{j-1} \Phi(\sigma_{j-1}) )\right)
\end{align*}
\noindent Again rearranging and taking matrix logarithm on both sides:
\begin{align}
\log\frac{\Tr A_j}{\Tr A_{j-1}} &\leq \mu\exp(\mu) \Tr \left(\rho_{j-1} \Phi(\sigma_{j-1})\right)   &   \log\frac{\Tr B_j}{\Tr B_{j-1}} &\leq -\mu\exp(-\mu) \Tr \left(\rho_{j-1} \Phi(\sigma_{j-1})  \right)  
\end{align}
\noindent Thus, an upper bound on $\Delta S(\rho^* \| \rho_j) + \Delta S(\sigma^* \| \sigma_j)$ is:
\begin{align*}
    \Delta S(\rho^* \| \rho_j) + \Delta S(\sigma^* \| \sigma_j) &\leq \mu\exp(\mu) \Tr \left(\rho_{j-1} \Phi(\sigma_{j-1}) \right) -\mu\exp(-\mu) \Tr \left(\rho_{j-1} \Phi(\sigma_{j-1})\right)
\end{align*}
\end{proof}

\subsection{Proof of Theorem \ref{thm:qre_mwu}}
In order to prove Theorem~\ref{thm:qre_mwu}, we utilize Corollary \ref{corollary:inequality} and take the limit of the inequalities as $\mu\to 0$.
\begin{proof}
Consider the MMWU update from time $j-1$ to $j$. As $\mu \to 0$, $\rho_j$ and $\sigma_j$ do not change more than $O(\mu)$. Indeed, since all payoffs in the game are bounded in the MMWU update from time $j-1$ to $j$, all entries in the numerators increase by at most $\exp(O(\mu)) = 1 + O(\mu)$. Likewise, the denominator is at least as large, but also upper bounded by the previous value of the denominator multiplied by $(1+O(\mu))$. Hence, every entry in the outputs $\rho_j$ and $\sigma_j$ are at most $O(\mu)$ from $\rho_{j-1}$ and $\sigma_{j-1}$. Moreover, we have that $Tr (\rho_j \Phi(\sigma_{j-1}))$, $Tr (\rho_{j-1}\Phi(\sigma_j))$ and $Tr (\rho_{j-1}\Phi(\sigma_{j-1}))$ are all within $O(\mu)$ of each other. 
Using Taylor expansion, both the upper bounds and lower bounds  given in Equations \ref{eqn:lowerboundsum},\ref{eqn:upperboundsum} are of the order of $O(\mu^2)$ and the theorem follows.

\end{proof}

\section{Proofs and Additional Results from Section \ref{sec:replicator}}
\label{appsecs:proofs_replicator}
In this section, we focus on the key results shown in Section \ref{sec:replicator}. In particular, we first present some connections between the quantum and classical replicator dynamics. Then, we show the comprehensive proof of Theorem \ref{thm:recurrence}, which relies on Lemmas \ref{lemma:volumeconservation} and \ref{lemma:boundedorbits}. In addition, the proofs of all other results can also be found in this section.
\subsection{Proof of Observation \ref{obs:equivalence}}
\begin{proof}
    It suffices to show that at the limit, the integrals $A(t)$ and $B(t)$ defined in \ref{eqn:A_t} and \ref{eqn:B_t} can be written in the discrete summation form as seen in Algorithm \ref{alg:quantum}. First we write the Riemann sum for $A(t)$ by taking $j$ infinitesimal intervals in time interval $[0,t]$, each of width $\mu$:
    \begin{align*}
        A(t) &= \int_0^t \Phi(\sigma(\tau))d\tau = \sum_{i=0}^{j-1} \left(\frac{t}{j}\right) \Phi(\sigma_i)
    \end{align*}
    However $t = \mu j$, so the above can be written as $\sum_{i=0}^{j-1} \mu \Phi(\sigma_i)$. Taking limit as $\mu \to 0$, it is clear to see that the matrix exponent of the continuous-time trajectory $A(t)$ is equal to its discrete-time counterpart $A_j$. Hence, the trajectories $\rho(t)$ and $\rho_j$ are equivalent at the limit. A similar argument holds for $\sigma(t)$ and $\sigma_j$ at the limit.
\end{proof}

\subsection{Connections between Quantum and Classical Replicator Dynamics}
In this section we show that in the special case where commutativity holds, the quantum replicator dynamics presented in Section \ref{sec:replicator} are equivalent to that of the classical replicator equations. Recall that we used the integral form of the replicator dynamics in the main text. However, typically in the classical case one can write the replicator dynamics in a form that represents the utilities obtained by each player: 
\[
\dot{x} = x\left(Ay - \left(x^\top A y\right)\Vec{\mathbb{1}}\right)
\]
where $x$ and $y$ are n-dimensional probability vectors representing the strategies of each player, $A$ is the payoff matrix and $\vec{\mathbb{1}}$ is the n-dimensional all-one vector. In the quantum setting, we require a few additional results and assumptions in order to write the Equations \ref{eqn:replicator} in the classical form. The following lemma arrives directly as a result of the series definition of matrix exponential.
\begin{lemma}
\label{lemma:exp_deriv}
If matrices $P$ and $\frac{dP}{dt}$ commute then
\[\frac{d(\exp(P))}{dt}= \frac{dP}{dt}\cdot \exp(P) = \exp(P)\cdot \frac{dP}{dt}\]
\end{lemma}
\begin{proof}
\begin{align*}
    \frac{d}{dt}(\exp(P)) &= \frac{d}{dt}\left[\mathbb{1} + P + \frac{1}{2} P^2 + \frac{1}{3!} P^3 + \cdots \right]\\
    &= \frac{dP}{dt} \left[\mathbb{1} + P + \frac{1}{2} P^2 + \cdots\right]\\
    &= \frac{dP}{dt}\cdot \exp(P) = \exp(P)\cdot\frac{dP}{dt} \qquad \mbox{(\text{by commutativity})}
\end{align*}
\end{proof}

By leveraging Lemma \ref{lemma:exp_deriv} and assuming that commutativity holds, we can state the following lemma:

\begin{lemma}
\label{lemma:replicator}
 If $\int_0^t \Phi(\sigma(\tau))d\tau$ \& $\Phi(\sigma(t))$, and $\int_0^t \Phi^\dagger(\rho(\tau))d\tau$ \& $\Phi^\dagger(\rho(t))$ commute, then the replicator system when applied to a two-player zero-sum quantum game is equivalent to:
\begin{eqnarray}
d\rho/dt &=& \rho [\Phi(\sigma)-\Tr(\rho \Phi(\sigma))\mathbb{1}]\label{eqn:replicator1_commute}\\
d\sigma/dt &=& \sigma [-\Phi^\dagger(\rho)+\Tr(\sigma \Phi^\dagger(\rho))\mathbb{1}]\label{eqn:replicator2_commute}
\end{eqnarray}
\end{lemma}

\begin{proof}
By definition we have that  $\rho=\exp(A)/\Tr(\exp(A))$ where $A = \int_0^t \Phi(\sigma(\tau))d\tau$.
By applying Lemma \ref{lemma:exp_deriv} we have that if $\int_0^t \Phi(\sigma(\tau))d\tau$ and $\Phi(\sigma(t))$ commute then
$d \exp(A)/dt= \exp(A)\Phi(\sigma)= \Phi(\sigma)\exp(A)$. 

 \begin{eqnarray*}
d\rho/dt &=& \frac{(d \exp(A)/dt) \Tr(\exp(A)) - \exp(A) d(\Tr(\exp(A)))/dt}{(\Tr \exp(A))^2}\\
  &=& \frac{(d\exp(A)/dt) \Tr(\exp(A))-\exp(A) \Tr(d\exp(A)/dt)}{(\Tr \exp(A))^2}\\
  &=& \frac{\exp(A)\Phi(\sigma) \Tr(\exp(A))- \exp(A) \Tr(\exp(A)\Phi(\sigma))}{(\Tr \exp(A))^2}\\
  &=& \rho [\Phi(\sigma)-\Tr(\rho \Phi(\sigma))\mathbb{1}]
\end{eqnarray*}
The proof in the case of $d\sigma/dt$ is similar.
\end{proof}

We can see that Equations \ref{eqn:replicator1_commute} and \ref{eqn:replicator2_commute} take a familiar form - indeed, the quantum replicator dynamics model the difference in utility obtained by each player compared to the average utility in much the same way as the classical case. In general, commutativity holds only in the case when all matrices $\Phi^\dagger(\rho(t))$ and $\Phi(\sigma(t))$ are diagonal, which is precisely the classical setting. The assumption that commutativity holds in the quantum setting is a very strong one, and in general commutativity does not hold once we enter the realm of quantum information.

Reminiscent of a similar result in classical evolutionary game theory, one can write a proof for the invariance of total entropy directly using the formulation of replicator dynamics. However given the commutativity assumption, this `proof' only holds in the case where we are writing classical probability vectors and game payoff matrices in the quantum notation. 

\begin{proof}[Proof of Theorem \ref{thm:qre_replicator} (Classical case)]
We will prove that  $d\mathrm{Tr}\big(\rho^* (\log \rho^*-\log \rho(t)) +  \sigma^* (\log \sigma^*-\log \sigma(t)) \big)/dt=0$.  To do that we will focus on the terms related to the first agent, i.e.,  $d\mathrm{Tr}\big(\rho^* (\log \rho^*-\log \rho(t))\big)/dt= -
d\mathrm{Tr}\big(\rho^* \log \rho(t))\big)/dt$.

\begin{align*}
d\mathrm{Tr}\big(\rho^* \log \rho(t))\big)/dt &= \mathrm{Tr}\big(\rho^* d(\log \rho(t))/dt\big)\\
  &= \mathrm{Tr}\big(\rho^* \rho(t)^{-1} d\rho(t)/dt \big)~~~~(\text{by Lemma}~ \ref{lemma:exp_deriv}\footnotemark)\\
  &=\mathrm{Tr}\big(\rho^* \rho^{-1}  \rho [\Phi(\sigma)-\Tr(\rho \Phi(\sigma))I] \big)~~~~(\text{by Lemma \ref{lemma:replicator}})\\
  &= \mathrm{Tr}\big(\rho^*  [\Phi(\sigma)-\Tr(\rho \Phi(\sigma))I] \big)\\
  &=\mathrm{Tr}\big(\rho^*  \Phi(\sigma)\big)-\Tr\big(\Tr(\rho \Phi(\sigma))\rho^*]\big)\\
  &= \mathrm{Tr}\big(\rho^*  \Phi(\sigma)\big) - \Tr(\rho \Phi(\sigma))~~~~\text{(Since } \Tr\big(\rho^*\big)=1)\\
  &= \mathrm{Tr}\big(\rho^*  \Phi(\sigma^*)\big) - \Tr(\rho \Phi(\sigma))
\end{align*}\footnotetext{Note here that $\rho(t)^{-1}$ exists since the exponential of a matrix is always an invertible matrix.}

The last line comes from the assumption that the equilibrium $(\rho^*,\sigma^*)$ is `fully mixed', i.e. the payoff of an agent when deviating from the Nash equilibrium to any other strategy remains exactly equal to her equilibrium payoff (e.g., Rock-Paper-Scissors). A similar analysis for the second agent results in:

\begin{eqnarray*}
d\mathrm{Tr}\big(\sigma^* \log \sigma(t))\big)/dt &=& - \Tr\big(\sigma^*  \Phi^\dagger(\rho^*)\big) + \Tr(\sigma \Phi^\dagger(\rho))
\end{eqnarray*}

These terms cancel out and the theorem follows.
\end{proof}

However in general we cannot use the above formulation, but rather need an argument tailored to the quantum setting using Theorem \ref{thm:qre_mwu} and Observation \ref{obs:equivalence}. The general proof of Theorem \ref{thm:qre_replicator} is presented here:
\subsection{Proof of Theorem \ref{thm:qre_replicator}}

In the proof of this theorem, we leverage Observation \ref{obs:equivalence} and Theorem \ref{thm:qre_mwu}. 
\begin{proof}
First we note that by expanding the limit definition of the derivative we obtain:
    \begin{align*}
        &\frac{d\big(S(\rho^* \| \rho(t))+S(\sigma^* \| \sigma(t))\big)}{dt} \\ =& \lim_{h\to 0}\frac{\big(S(\rho^* \| \rho(t+h))+S(\sigma^* \| \sigma(t+h))\big)-\big(S(\rho^* \| \rho(t))+S(\sigma^* \| \sigma(t))\big)}{h}
    \end{align*}

Here, we apply Observation \ref{obs:equivalence} in the following manner: If we perform the substitution $j\to t+h$, $j-1 \to t$ and $\mu \to h$, we obtain a 
reformulation of the same limit in the language of
 stepsizes of duration $\mu$. As $\mu \to 0$,
 the difference between the continuous time flow solution and its discrete-time Euler approximation is of order $O(\mu^2)$. Thus, we can compute this limit evaluated along the points of the discrete-time MMWU trajectory.
 This limit by Theorem 3.5 is equal to zero. 

    


    
\end{proof}

\subsection{Proof of Proposition \ref{prop:map1}}
\begin{proof}
We first consider the definition of $\rho'(t)$.
\begin{align*}
    \rho'(t) &= \frac{\exp(A'(t))}{\Tr(\exp(A'(t)))}\\
    &= \frac{\exp\left(\int_0^t \Phi(\sigma(\tau))d\tau - (v^\dag A(t) v)\mathbb{1}\right)}{\Tr\left( \exp\left(\int_0^t \Phi(\sigma(\tau))d\tau - (v^\dag A(t) v)\mathbb{1}\right)\right)}\\
    &= \frac{\exp\left(\int_0^t \Phi(\sigma(\tau))d\tau\right)\cdot \exp\left(- (v^\dag A(t) v)\mathbb{1}\right)}{\Tr\left( \exp\left(\int_0^t \Phi(\sigma(\tau))d\tau \right)\cdot\exp\left(- (v^\dag A(t) v)\mathbb{1}\right)\right)}
\end{align*}
Note that the denominator can be written as the trace of a matrix where each diagonal entry is the corresponding value of $A$ multiplied by $\exp(-v^\dag A(t) v)$. Thus, the above can be rewritten as
\begin{align*}
    \rho'(t) &= \frac{\exp(A(t)) \cdot \exp\left(- (v^\dag A(t) v)\mathbb{1}\right)}{\exp(-v^\dag A(t) v)\cdot\Tr\left(\exp(A(t))\right)}\\
    &= \frac{\exp(A(t))}{\Tr(\exp(A(t)))}\cdot\mathbb{1}\\
    &= \rho(t)
\end{align*}
We can perform similar computations to show the same holds true for $\sigma'(t)$ and $\sigma(t)$.
\end{proof}

\subsection{Proof of Proposition \ref{prop:diffeomorphic}}
\begin{proof}
Consider the map $f$ between $A'$ and $\rho$. We have shown in Proposition \ref{prop:map1} that the map in one direction is $\rho = f(A') =\frac{\exp(A')}{\Tr(\exp(A'))}$. This is continuously differentiable. Now let us consider the inverse map $f^{-1}$. In particular, we show that the inverse mapping exists and is equal to $A' = f^{-1}(\rho) = \log(\rho)-(v^\dag \log(\rho) v)\mathbb{1}$.

Indeed,
\begin{align*}
    f(f^{-1}(\rho)) &= \frac{\exp(\log(\rho)-(v^\dag \log(\rho) v)\mathbb{1})}{\Tr(\exp(\log(\rho)-(v^\dag \log(\rho) v)\mathbb{1}))}\\
    &= \frac{\exp\left(\log(\rho)\right)\cdot\exp\left(-(v^\dag \log(\rho) v)\mathbb{1}\right)}{\Tr\left(\rho\cdot\exp\left(-(v^\dag \log(\rho) v)\mathbb{1}\right)\right)}\\
        &= \frac{\rho\cdot\exp\left(-(v^\dag \log(\rho) v)\mathbb{1}\right)}{\exp\left(-(v^\dag \log(\rho) v)\right)}\\
        &= \rho\cdot\mathbb{1} = \rho
\end{align*}
where we have used the fact that $\Tr(\rho) = 1$.
Furthermore, note that the inverse mapping $f^{-1}(\rho) = \log(\rho)-(v^\dag \log(\rho) v)\mathbb{1}$ is smooth.

Hence, since $f$ is a bijection and the inverse map is differentiable, $A'(t)$ and $\rho(t)$ are diffeomorphic to each other. Similarly, $B'(t)$ and $\sigma(t)$ are diffeomorphic to each other.
\end{proof}

Lemmas \ref{lemma:volumeconservation} and \ref{lemma:boundedorbits} are crucial to proving our recurrence result. Typically, volume conservation is not difficult to prove, since the canonical transformation guarantees that the dynamical system becomes separable, which then allows for a simple application of Liouville's theorem. However, proving that the orbits of the dynamical system are bounded is more challenging. We carefully consider the definitions of $A'(t)$ and $B'(t)$ and utilize properties regarding their eigenvalues to show this property.

\subsection{Proof of Lemma \ref{lemma:volumeconservation}}
\begin{proof}
    Let $\Psi_{A'}$ denote the flow of $A'(t)$ and $\Psi_{B'}$ denote the flow of $B'(t)$. First let us consider $A'(t)$.
    \begin{align*}
       \dot{A}' = \frac{dA'}{dt} &= \Phi(\sigma(t)) - (v^\dag\Phi(\sigma(t))v)\mathbb{1}\\
        &= \Phi\left(\frac{\exp(B(t))}{\Tr(\exp(B(t)))}\right)-\left(v^\dag\Phi\left(\frac{\exp(B(t))}{\Tr(\exp(B(t)))}\right)v \right)\mathbb{1}
    \end{align*}
    
    Now, we derive the first order partial-derivatives of $F(A')$ with respect to $A'(t)$:
    \begin{align*}
        \frac{\partial F(A')}{\partial A'} &= \frac{\partial}{\partial A'} \bigg(\Phi\left(\frac{\exp(B(t))}{\Tr(\exp(B(t)))}\right)-\left(v^\dag\Phi\left(\frac{\exp(B(t))}{\Tr(\exp(B(t)))}\right)v\right) \mathbb{1} \bigg)\\
        &= \frac{\partial}{\partial A'} \bigg(\Phi\left(\frac{\exp(B'(t))}{\Tr(\exp(B'(t)))}\right)-\left(v^\dag\Phi\left(\frac{\exp(B'(t))}{\Tr(\exp(B'(t)))}\right)v\right) \mathbb{1} \bigg)\\
        &= 0
    \end{align*}
    Clearly, the partial derivative $\frac{\partial F(A')}{\partial A'}$ only depends on the value of $B'(t)$ and not $A'(t)$. This implies that the vector field is separable, and so the first order partial derivative with respect to $A'(t)$ is zero. By a similar argument, the partial derivative $\frac{\partial F(B')}{\partial B'}$ is also zero. By definition, the diagonal of the Jacobian matrix describing the vector fields is zero, and hence the divergence (which is the trace of the Jacobian) is zero as well. We can then directly apply Liouville's theorem to conclude that the flows $\Psi_{A'}$ and $\Psi_{B'}$ are volume preserving. 
\end{proof}

\subsection{Proof of Lemma \ref{lemma:boundedorbits}}
\begin{proof}
     First, $\rho(0) = \frac{\exp(A(0))}{\Tr(\exp(A(0)))}$ (resp. $\sigma(0)$) is bounded away from the boundary, since we assumed that 
    since $A(0)$ and $B(0)$ are finite matrices.  Moreover, $A'(0)$ and $B'(0)$ are also finite. 
    The sum of quantum relative entropies at time 0 is thus:
    \[
    S(\rho^* \| \rho(0))+ S(\sigma^* \| \sigma(0)) = \Tr(\rho^*(\log\rho^*-\log\rho(0))) + \Tr(\sigma^*(\log\sigma^*-\log\sigma(0)))
    \]
    This sum is finite since $\rho(0), \sigma(0)$ have full support. Indeed, the support of $\rho^*$ and $\sigma^*$ are contained in the support of $\rho(0)$ and $\sigma(0)$ respectively.
    By Theorem \ref{thm:qre_replicator}, the sum of quantum relative entropies $S(\rho^* \| \rho(t))+ S(\sigma^* \| \sigma(t))$ is bounded above by a fixed value $C$ for all time $t$. 
    Now, let $\lambda_{\min}(\rho^*)$ denote the minimum eigenvalue of $\rho^*$. We have that $\rho^* \geq \lambda_{\min}(\rho^*)\mathbb{1}$. Note the following:
    \begin{align*}
        \Tr(\rho^*(\log\rho^*))-\Tr(\rho^*(\log\rho(t))) &< C\\
        -\Tr(\rho^*(\log\rho(t))) &< D\\
        -\Tr(\lambda_{\min}(\rho^*)\mathbb{1}(\log\rho(t))) &< D
    \end{align*}
    where C and D are positive and finite real numbers. Using the fact that the trace of a matrix is basis independent, we use a basis where $\rho(t)$ is diagonal. Here, the minimum eigenvalue of $\rho(t)$ cannot go to zero, otherwise  $-\Tr(\lambda_{\min}(\rho^*)\mathbb{1}(\log\rho(t)))$ goes to $+\infty$, a clear contradiction to the inequality above. A similar argument holds for $\sigma(t)$. This is equivalent to saying that all the eigenvalues of $\rho(t)$ and $\sigma(t)$ lie in the interval $[\epsilon, 1-\epsilon]$ for some small $\epsilon>0$. 
    
    Note that the value of $A'_{1,1}(t)$ is always zero by design. We also know by construction that $A'(t)$ is Hermitian, since $\sigma(t)$ is Hermitian, $\Phi$ constitutes a hermicity preserving map and the integral of Hermitian matrices remains Hermitian. Hence, the smallest eigenvalue of $A'(t)$, $\lambda_{\min}(A'(t))$ is upper bounded by the smallest element on the diagonal of $A'(t)$, namely $\lambda_{\min}(A'(t)) \leq 0$, since $A'_{1,1}(t) = 0$. If the largest eigenvalue $\lambda_{\max}(A'(t))$ goes to $+\infty$, then the corresponding smallest eigenvalue of $\rho(t) = \frac{\exp(A'(t))}{\Tr(\exp(A'(t)))}$ goes to $0$. 
    Likewise if the smallest eigenvalue $\lambda_{\min}(A'(t))$ goes to $-\infty$, the smallest eigenvalue of $\rho(t)$ goes to $0$ as well.
    We have previously shown that this cannot occur, since the eigenvalues of $\rho(t)$ are all bounded away from zero. Hence, it follows that all eigenvalues of $A'(t)$ are bounded away from infinity. 
    
    Now, consider the scalar $v^\dag A'(t) w$ for arbitrary bounded vectors $v$ and $w$. $A'(t)$ is Hermitian, so by the spectral theorem it has an orthonormal eigenbasis. Now we express $v$ and $w$ in terms of this $n$-dimensional eigenbasis (which we denote by $\mathbb{e}$):
    \begin{align*}
        v = \sum_{i=1}^n a_i \mathbb{e}_i \quad\text{and}\quad w = \sum_{i=1}^n b_i \mathbb{e}_i 
    \end{align*}
    where $a_i$ and $b_i$ are scalars. Thus,
    \begin{align*}
        v^\dag A'(t) w &= \left(\sum_{i=1}^n a_i \mathbb{e}_i\right)^\dag A'(t) \left(\sum_{i=1}^n b_i \mathbb{e}_i\right)\\
        &= \sum_{i=1}^n a_i b_i \lambda_i
    \end{align*}
    where $\lambda_i$ are the eigenvalues of $A'(t)$ corresponding to $\mathbb{e}_i$. Here $a_i$, $b_i$ and $\lambda_i$ are all bounded, so it follows that all elements of matrix $A'(t)$ are bounded away from infinity as well. Likewise, the elements of matrix $B'(t)$ remain bounded from infinity and thus the dynamical system described by $A'(t)$ and $B'(t)$ have bounded orbits.

\end{proof}

\section{Experiments}
\label{appsecs:experiments}
In order to formalize a method for simulating and understanding quantum games, we propose a method to generate complex R matrices with the same eigenvalues as the values of the classical payoff matrix. This means that for any standard eigenbasis of a classical game, we create a new, complex eigenbasis where the eigenvalues are the same, but the eigenvectors are complex. Intuitively, this can be viewed as mapping a classical game matrix to the Hilbert space, effectively allowing generalizations of classically studied games to the semi-definite context. We introduce the following lemma, which describes a method to obtain a basis transformation to Hilbert space. 

\begin{lemma}\label{lemma:hadamardtransform}
    For any $n\times n$ two player zero-sum game with classical payoff matrix $P$, let $\lambda_{i,j} = P_{i,j}$ and let $V$ and $W$ be unitary matrices of size $n\times n$.
    Then, the matrix $R$ defined as
    \begin{equation}\label{eqn:basistransform}
        R = \sum_{i,j} \lambda_{i,j} (V_i\otimes W_j) (V_i^\dag \otimes W_j^\dag)
    \end{equation}
    produces a transformation from the classical payoff space $\mathcal{P}$ to the density matrix space $\mathcal{C}$. In particular, $R$
    is a complex $n^2\times n^2$ matrix which satisfies the following properties:
    \begin{itemize}
        \item The matrix $R$ is Hermitian.
        \item The eigenvalues of $R$ correspond to the elements of the classical payoff matrix $P$.
    \end{itemize}
\end{lemma}
\begin{proof}
    First note that Equation \ref{eqn:basistransform} is equivalent to taking the following basis transformation:
    \[
    \sum_{i,j} \lambda_{i,j} (V_i\otimes W_j) (V_i^\dag \otimes W_j^\dag) = (V\otimes W){\Lambda} (V^\dag\otimes W^\dag)
    \]
    where $\Lambda = \begin{bmatrix}
    \lambda_{1,1} & \hdots & 0 \\
    \vdots & \ddots & \vdots\\
    0 & \hdots & \lambda_{n,n}
    \end{bmatrix}$ is the $n^2 \times n^2$ diagonal matrix formed by placing the elements of classical payoff matrix $P$ on the diagonal. Moreover, $V$ and $W$ are both unitary, so the tensor products $V\otimes W$ and $V^\dag\otimes W^\dag$ are also unitary. This holds because:
    \begin{align*}
        (V\otimes W)(V \otimes W)^\dag &= (V\otimes W)(V^\dag \otimes W^\dag)\\
        &= VV^\dag\otimes WW^\dag\\
        &= \mathbb{1}\otimes \mathbb{1} = \mathbb{1}
    \end{align*}
    
    This implies that $R$ can be written as a unitary diagonalization of the form $PDP^\dag$, where $P=(V\otimes W)$ and $D = \Lambda$. Hence $R$ is Hermitian and has real eigenvalues by the spectral theorem. In addition, by construction the eigenvalues of $R$ are the elements of the classical payoff matrix $P$. 
\end{proof}
    

Moreover, the initial conditions of a simulation of a two-player quantum zero-sum game can be similarly obtained by using the same orthonormal bases as in Lemma \ref{lemma:hadamardtransform} and applying a change of basis to the initial conditions in the standard basis. 

\subsection{Additional Experimental Details (Matching Pennies)}
In the main text, we showed several simulations that illustrate the theoretical results in Sections \ref{sec:mwu} and \ref{sec:replicator}. Here we present detailed descriptions of the simulations, as well as additional experiments which serve to corroborate auxiliary results in the paper.

To generate Figure \ref{fig:quantum_rps_entropy}, we used a modified version of Algorithm \ref{alg:quantum}, where instead of constant step-size $\mu$ we define $\mu = \log(1+\frac{1}{t^{a}})$, where $a$ is an exponent which determines the rate of decrease of the step-size, and $t$ is the time-step of the simulation. Intuitively, this means that for higher values of $a$, the step-sizes go to zero at a faster rate. Empirically, we observe that if the exponent is greater than $\frac{1}{2}$, then we can clearly see the phenomenon described in Theorem \ref{thm:qre_mwu}, since the quantum relative entropy in the system increases when the trajectories move away from the fully mixed Nash equilibrium. 

In the case of replicator dynamics, we show experimentally (Figure \ref{fig:qre1}) that the sum of quantum relative entropies for both players is a constant of motion as shown in Theorem \ref{thm:qre_replicator}. We first use Lemma \ref{lemma:hadamardtransform} to obtain the $R$ matrix of a quantum game given some complex basis and classical payoff matrix (Matching Pennies, for instance). Then, we run the discrete Algorithm \ref{alg:quantum} with decreasing step-size to compute the interior $\epsilon$-approximate Nash of the quantum game. The replicator equations are then solved given a set of initial conditions for each player. Thereafter, we compute the quantum relative entropy between each player's strategy at each discretized time step and the $\epsilon$-approximate Nash equilibrium.

Notice that although the quantum relative entropies of the first player (pink) and second player (purple) are oscillating over time, their sum remains approximately constant. 
\begin{figure}[H]
    \centering
      \includegraphics[width=.75\linewidth]{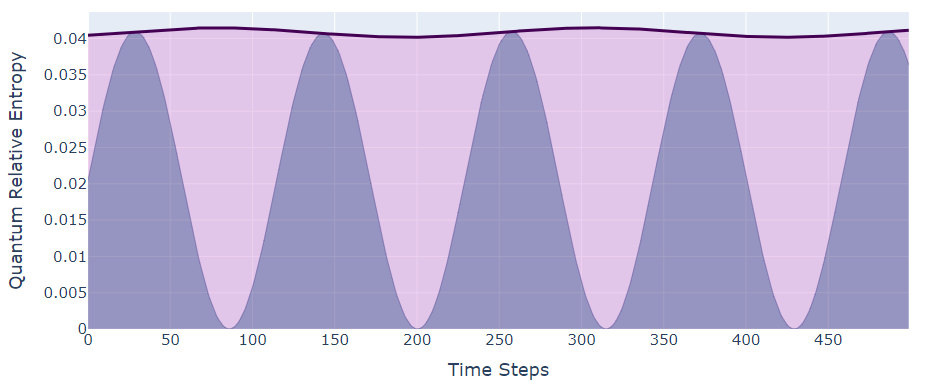}
    \caption{Approximately constant sum of quantum relative entropy values for two-player quantum Matching Pennies game.}
    \label{fig:qre1}
\end{figure}

Finally, in order to generate Figure \ref{fig:bloch_sphere}, we perform a similar simulation to the above. However, we also transform the trajectories of each player using the Pauli matrices, which are given as:
\begin{equation}
    \sigma_{x} = \begin{pmatrix}
    0 & 1\\
    1 & 0
    \end{pmatrix} \qquad \sigma_{y} = \begin{pmatrix}
    0 & -i\\
    i & 0
    \end{pmatrix} \qquad \sigma_{z} = \begin{pmatrix}
    1 & 0\\
    0 & -1
    \end{pmatrix}
\end{equation}
The expectations of the density matrices under this transformation (commonly referred to as the Bloch vector) produce the coordinates on the $x$, $y$ and $z$ axes within the Bloch sphere. We plot these coordinates for each time step and hence obtain the trajectories for each player in the Bloch sphere. 

\subsection{Larger-Scale Experiments}\label{appsecs:largerscale}
In the main text, we study a single-qubit system in the form of a Matching Pennies game. However, our results extend to games with larger strategy sets, some of which can be interpreted as multi-qubit systems. We study a $16\times16$ game with a payoff matrix in $[-1,1]$ which can be played via 2 qubits. Moreoever, this game is selected because it has an interior (and uniform) Nash equilibrium. Similarly, we generate an $64\times64$ generalization of the $16\times 16$ game to represent an example of a 3-qubit game with interior Nash. In Figure \ref{fig:rps_recurrence}, we plot the Frobenius norm between the initial condition and the system state when updated using quantum replicator dynamics (Equations \ref{eqn:replicator}). Notice that despite the relatively erratic behaviour of the strategies, they eventually return to the initial condition, implying that recurrence holds (Theorem \ref{thm:recurrence}).

\begin{figure}[ht]
    \centering
    \begin{minipage}{.9\linewidth}
      \centering
      \includegraphics[width=.95\linewidth]{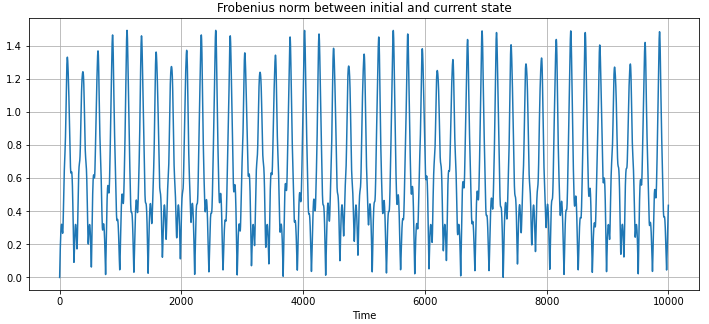}
      2-qubit system
    \end{minipage}\par\bigskip
    \begin{minipage}{.9\linewidth}
      \centering
      \includegraphics[width=.95\linewidth]{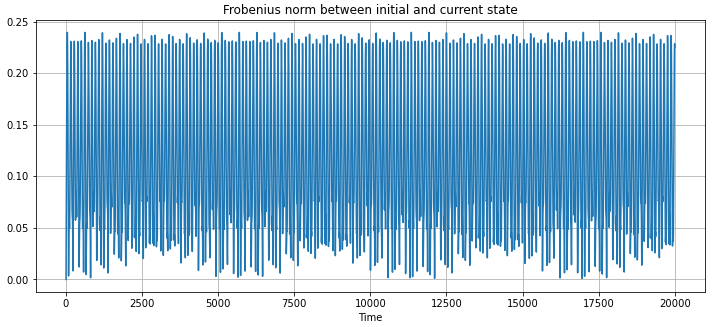}
      3-qubit system
    \end{minipage}
    \caption{Frobenius norm between the initial condition and current system state for 2 and 3 qubit systems with interior Nash equilibrium. Note that at approx. 7200 iterations the 2-qubit system returns arbitrarily close to 0, and the same occurs at approx. 12000 iterations for the 3-qubit system.}
    \label{fig:rps_recurrence}
\end{figure}

Furthermore, to confirm the volume conservation property shown in Theorem \ref{thm:qre_replicator} we perform the same simulation as in Figure \ref{fig:qre1} for the same 2 and 3 qubit systems as above, obtaining Figure \ref{fig:rps_qre}. However, a minor note is that we perform these simulations without the basis transform of Lemma \ref{lemma:hadamardtransform} to avoid potential numerical errors since the matrices are of larger dimension. The initial conditions are random density matrices, as before. Note that the sum of quantum relative entropies between each player and the interior Nash equilibrium is constant as expected from our theoretical results. 

\begin{figure}[ht]
    \centering
    \begin{minipage}{.9\linewidth}
      \centering
      \includegraphics[width=.95\linewidth]{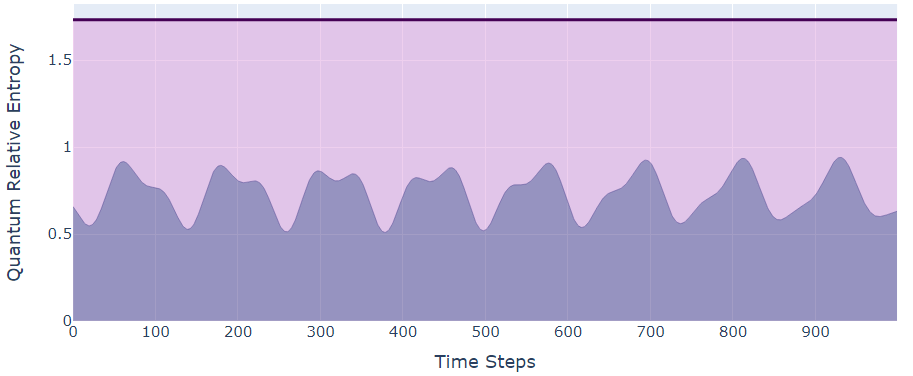}
      2-qubit system
    \end{minipage}\par\bigskip
    \begin{minipage}{.9\linewidth}
      \centering
      \includegraphics[width=.95\linewidth]{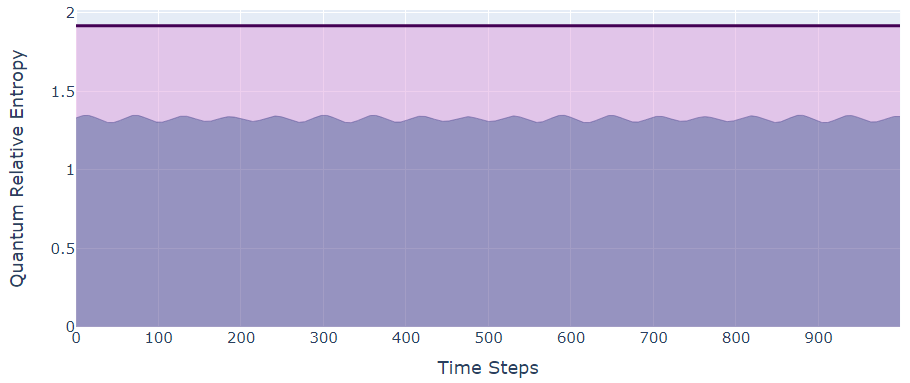}
      3-qubit system
    \end{minipage}
    \caption{Sum of quantum relative entropies of each player to the Nash equilibrium for 2 and 3 qubit systems with interior Nash.}
    \label{fig:rps_qre}
\end{figure}


\subsection{Reproducibility Details}
All experiments performed for this work were done using Python 3.7 and have been compiled into a Jupyter notebook for ease of viewing. Running the code requires basic scientific computing packages such as NumPy, SciPy and Cython, as well as data visualization packages such as Matplotlib and Plotly. In addition, we use QuTip \cite{johansson2012qutip}, a quantum simulation package. Most of our code has been edited such that it can be easily executed on a standard computer in a matter of minutes.

\end{appendices}

\end{document}